\documentclass[11pt]{article}
\usepackage[latin,english]{babel}
\usepackage[T1]{fontenc}
\usepackage{eufrak}
\usepackage{geometry}
\usepackage{bbm}
\usepackage{amsbsy}
\usepackage{makeidx}
\usepackage{mathtools}
\usepackage{amssymb, amsthm}
\usepackage{bm}
\usepackage{mathrsfs, dsfont}
\usepackage{cases}
\usepackage{braket}
\usepackage[toc,page]{appendix}
\usepackage{dirtytalk}
\usepackage{pdfsync}
\usepackage{graphicx}
\usepackage{epstopdf}
\usepackage{todonotes}
\usepackage{comment}
\usepackage{tikz}
\usepackage[normalem]{ulem}
\usetikzlibrary{shapes,decorations,calc,fit}
\usepackage{graphicx}
\usepackage{fancyhdr}
\usepackage{math}
\usepackage{cite}
\usepackage{color}
\usepackage{subcaption}
\usepackage{multicol}
\usepackage[permil]{overpic}
\newcommand{\pt}{\partial}
\newcommand{\bbs}[1]{\boldsymbol{#1}}

\usepackage{xcolor}
\definecolor{lgray}{RGB}{240,240,240}
\definecolor{webgreen}{rgb}{0,.5,0}
\definecolor{webbrown}{rgb}{.6,0,0}
\definecolor{RoyalBlue}{cmyk}{1, 0.50, 0, 0}
\usepackage[colorlinks=true, breaklinks=true, urlcolor=webbrown, linkcolor=RoyalBlue, citecolor=webgreen, backref=page]{hyperref}

\numberwithin{equation}{section}
\allowdisplaybreaks

\begin{document}

\title{Soliton Synchronization with Randomness: \\ Rogue Waves and Universality} 
\author{Manuela Girotti\footnote{Corresponding author: manuela.girotti@emory.edu}$\ ^{,a}$, Tamara Grava$^{b,c}$, Robert Jenkins$^d$, Guido Mazzuca$^e$, \\ Ken McLaughlin$^e$, Maxim Yattselev$^f$}

\date{{\small $^a$ Emory University, $^b$ SISSA, $^c$ University of Bristol, $^d$ University of Central Florida, \\ $^e$ Tulane University, $^f$ Indiana University Indianapolis}}

	\maketitle

\begin{abstract}
    We consider an $N$-soliton  solution of the focusing nonlinear Schr\"{o}dinger equations.  We give conditions for the synchronous collision  of these $N$ solitons. When the solitons velocities are well separated  and the solitons have equal amplitude,  we  show that the local  wave profile at the collision point   scales as  the   \( \sinc(x) \) function. We show that this behaviour persists when the amplitudes of the solitons are  i.i.d. sub-exponential random variables.  Namely the {central collision peak} exhibits \emph{universality}: its spatial profile converges   to the $\sinc(x)$ function, independently of the distribution. We derive Central Limit Theorems for the fluctuations of the profile in the near-field regime (near the collision point) and in the far-regime.
\end{abstract}


\section{Introduction and Main Results}

We consider the focusing Nonlinear Schr\"{o}dinger (fNLS) equation in $1+1$ dimensions
\begin{equation}\label{eq:nls}
    \ii \psi_t + \frac{1}{2}\psi_{xx} + |\psi|^2 \psi = 0, \quad (x,t) \in \R\times [0,+\infty)\ .
\end{equation}
This equation has countless applications both in physics and engineering. It serves as a model of nonlinear waves: in particular, water waves of small amplitude over
infinite depth \cite{Zakharov68} and finite depth \cite{BenneyRoskes69, HasimotoOno72}, as well as almost monochromatic waves in a weakly nonlinear dispersive medium \cite{BenneyNewell67,CalogeroEckhaus87} and rogue waves  \cite{Chabchoubetal11,Peregrine83}.  It also appears in the study of the propagation of signal in fiber optics \cite{Suret2011, Chiaoetal64, Talanov64}, plasma of fluids \cite{Zakharov71},  Bose–Einstein condensations  \cite{PethickSmith02}, to mention a few. 
In this manuscript  we analyze the behaviour of soliton interactions.  In particular
\begin{itemize}
    \item 
we identify the 
 soliton  parameters that maximizes the  amplitude of the wave profile at the interaction time;
 \item we prove that a train  of   solitons of equal amplitudes and  well separated velocities  has a distinguished  collision profile given by the $sinc$  function;
\item  we prove the universality of  the $sinc$  profile   by showing that it surprisingly persists when the soliton amplitudes are  sampled from a probability  distribution, while  the velocities  are well separated, and deterministic. 
\end{itemize}
We further   confirm   the well known fact  that fast solitons interact linearly at leading order, and  we provide an explicit expression of the sub-leading (nonlinear) corrections that is instrumental to obtain our results.
Below we explain in detail our results.

\subsection{Deterministic  Soliton Solutions}

The fNLS equation \eqref{eq:nls} is an example of an integrable equation that admits {\em soliton solutions}. The simplest of these is the family of one-soliton solutions 
\begin{equation}
\label{one_soliton}
\begin{cases}
\psi(x,t;z,c) = \mu \sech\big(\mu(x-x_0-v t  )\big) \ee^{\ii \lp x v - \tfrac{t}{2}(v^2-\mu^2) - \phi_0- \pi \rp}, \\
z = \tfrac12(-v+ \ii \mu) \in \C^+, \quad c = \ii \mu \ee^{\mu x_0 + \ii \phi_0} \in \C^*=\C\setminus\{0\}\,,
\end{cases}
\end{equation}
where $\C^+$ denotes the complex upper half-plane. 
Each such solution describes a localized traveling wave with velocity $v$ and maximum amplitude $\mu$. Given $2N$ complex constants, which we call \emph{reflectionless scattering data},
\begin{equation}\label{eq:soliton parameters}
    \{(z_k, c_k)\}_{k=1}^N \in \mathbb{C}^+ \times \C^*,  \qquad z_k = \tfrac12(-v_k+ \ii \mu_k), \quad c_k = \ii \mu_k \ee^{\mu_k x_k + \ii \phi_k},
\end{equation}
equation \eqref{eq:nls} also admits an $N$-soliton solutions, which we denote by \( \psi_N(x,t) \), whose absolute value has the following determinantal representation  
\begin{equation}\label{eq:nls Kay-Moses}
    |\psi_N(x,t)|^2 = \partial_{xx} \log \det \big( {\boldsymbol I}_N + {\boldsymbol \Phi}_N(x,t) \overline{\boldsymbol \Phi_N(x,t)} \big),
\end{equation}
where $\boldsymbol\Phi_N(x,t)$ is the $N \times N$ matrix with entries
\begin{equation}
    \left[ \boldsymbol\Phi_N(x,t) \right]_{jk} = \frac{ \sqrt{c_j \overline c_k} \ee^{2\ii(\theta(z_j,x,t)-\theta(\overline z_k,x,t))}}{\ii(z_j-\overline z_k)},
    \qquad
    \theta(z,x,t) = x z + t z^2.
\end{equation}
Formula \eqref{eq:nls Kay-Moses} for the wave field of an $N$-soliton solution can become quite involved, as $N$ gets large. However,
the scattering theory for the reflectionless fNLS provides a general upper bound on the wave amplitude in terms of the scattering data: 
\begin{equation}\label{eq:upper_bound_result}
    | \psi_N(x,t)| \leq \sum_{k=1}^N \mu_k.
\end{equation}
In fact, this upper bound is tight, and the following proposition gives a precise description of how it can be realized.    
\begin{proposition}
\label{thm:darboux_intro}
The  $N$-soliton solution $\psi_N(x,t)$ of \eqref{eq:nls} described by reflectionless scattering data \eqref{eq:soliton parameters} with 
    norming constant 
    \begin{equation}\label{eq:max_norming_const}
    \displaystyle c_k =  \frac{\ee^{-2\ii z_k x_0 -2\ii z_k^2 t_0}}{B'(z_k)} \ , \quad  \forall \, k=1,\ldots, N\ , \qquad
    B(z) := \prod_{k=1}^N\frac{z-z_k}{z-\overline z_k}\ ,
\end{equation}
    realizes the upper bound in \eqref{eq:upper_bound_result} at \( x=x_0 \) and $t =t_0$, namely,
    \begin{equation}
    \label{upper bound}
        |\psi_N(x_0,t_0)|= \sum_{k=1}^N \mu_k\,.
    \end{equation}
\end{proposition}

Additionally, it is well known, see \cite{Jenkins2018}, that whenever the velocities $v_k = -2\Re(z_k)$ are distinct, such a solution resolves asymptotically in the large time limit to the sum of one-soliton solutions
\begin{equation}
    \psi_N(x,t) = \sum_{k=1}^N \psi(x-x^\pm_k,t;z_k,c_k) e^{\ii \phi_k^\pm} + \bigo{e^{-\kappa|t|}}, \qquad t \to \pm \infty,
\end{equation}
where $\kappa >0$ depends on the eigenvalues $z_k$, $k=1,\dots N$, and the asymptotic phase shifts are given by
\begin{equation}\label{eq:NLS phase shifts}
    x_k^+ - x_k^- = \frac{2}{\mu_k} \sum_{j \neq k}^N \sgn(v_j-v_k) \log \left| \frac{ z_k-z_j}{z_k - \overline z_j} \right|, 
    \qquad
    \phi_k^+ - \phi_k^- = 2\sum_{j \neq k}^N \sgn(v_j-v_k) \arg\left( \frac{ z_k-z_j}{z_k - \overline z_j} \right).
\end{equation}
This formula is the justification for the statement that soliton collisions are elastic (see again \cite{Jenkins2018}), in the sense that the physical characteristics (velocity and amplitude) of the solitons are invariant before and after the collision; the only effect of the collision is the induced phase shifts.

On the other hand, for intermediate times the interactions of the "individual" solitons can generate very large wave peaks as supported by Proposition~\ref{thm:darboux_intro}. 
In this setting, the solitons interact almost linearly, indeed at the level of \eqref{eq:NLS phase shifts} one sees that 
 \begin{equation}
  \frac{ z_k - z_j}{z_k - \overline z_j} = 1 + \frac{2\ii \mu_j}{v_k - v_j} + \bigo{(v_k -v_j)^{-2}} ,  
\end{equation}
so that, when the soliton velocities are well separated, the phase shifts due to the pairwise soliton interactions become small. 
Our first main result gives a precise characterization of the asymptotic linearity of the soliton interactions when their velocities are well-separated. Moreover, it provides an explicit first correction to the leading order linearity with bounds on the higher order corrections.

\begin{theorem}
\label{thm:1}
Given $N \in \mathbb{N}$, consider the  $N$-soliton solution $\psi_N(x,t)$ of \eqref{eq:nls} corresponding to the reflectionless spectral data \eqref{eq:soliton parameters} with distinct velocities. Define 
\[
\Delta := \min_{j\neq k}|v_j-v_k|>0.
\]
Let \( \boldsymbol\mu=(\mu_1,\mu_2,\ldots,\mu_N) \) be the vector of soliton amplitudes. Then there exists $C_*>0$  such that for all $\Delta>C_*\|\boldsymbol\mu\|_\infty$ it holds that
\begin{multline} 
\label{linear interaction}
    \psi_N(x,t) 
   = \sum_{k=1}^N \psi^{(k)}(x,t) + \\
  \qquad  \frac{1}{2\ii } \sum_{j=1}^N\sum_{\substack{ k=1 \\ k \neq j}}^N
   \left[  \frac{\psi^{(k)}(x,t){m^{(j)}(x,t)}}{z_j - \overline z_k} - \frac{m^{(k)}(x,t) \psi^{(j)}(x,t)}{\overline z_j - \overline z_k}
    \right]
    + \bigo{\frac{\|\boldsymbol\mu\|_\infty \| \boldsymbol\mu\|_2^2}{\Delta^{2}}},
\end{multline}
where \( \psi^{(k)}(x,t) =\psi(x,t;z_k,c_k)\) is the one-soliton solution \eqref{one_soliton} and \( m^{(k)}(x,t) = \int_x^\infty |\psi^{(k)}(s,t)|^2 \di s \). The error bounds are uniform for all $(x,t)\in \R^2$ and depend on $N$ only through the norms $\| \boldsymbol\mu \|_\infty := \sup_{1\leq k \leq N} |\mu_k|$ and $\| \boldsymbol{\mu}\|_2^2 := \sum_{k=1}^N | \mu_k |^2 $.
\end{theorem} 

\begin{remark}
\label{error terms}
    The proof of Theorem~\ref{thm:1} shows that the first correction term, the double sum in \eqref{linear interaction}, is of order \( \|\boldsymbol{\mu}\|_2^2/\Delta \). In particular, if the magnitudes \( \mu_k \) are uniformly (independently of \( N\)) bounded above and away from zero, then \( \|\boldsymbol{\mu}\|_2^2 \sim N \). In this case, the leading term in \eqref{linear interaction} can be of order \( N \) (see \eqref{upper bound} and \eqref{sinc profile} further below), the first correction term is of order \( N/\Delta \), while the remaining error terms are of order \( N/\Delta^2 \).
\end{remark}

Theorem~\ref{thm:1} shows that for well separated velocities the phase shifts induced by soliton interactions become negligible. In what follows, we tune the discrete spectral data so that: (1) the  solitons all collide at a fixed point in space time $(x_0,t_0)$, that in view of the Galilean invariance of the fNLS equation, we assume without loss of generality to be \( (0,1) \); and (2) the soliton phases are tuned so that the maximum of the solution at collision time is exactly equal to the sum of the soliton amplitudes, as proven in Proposition \ref{thm:darboux_intro}.

\begin{figure}[t!]
\centering
\begin{overpic}[width=.9\linewidth]{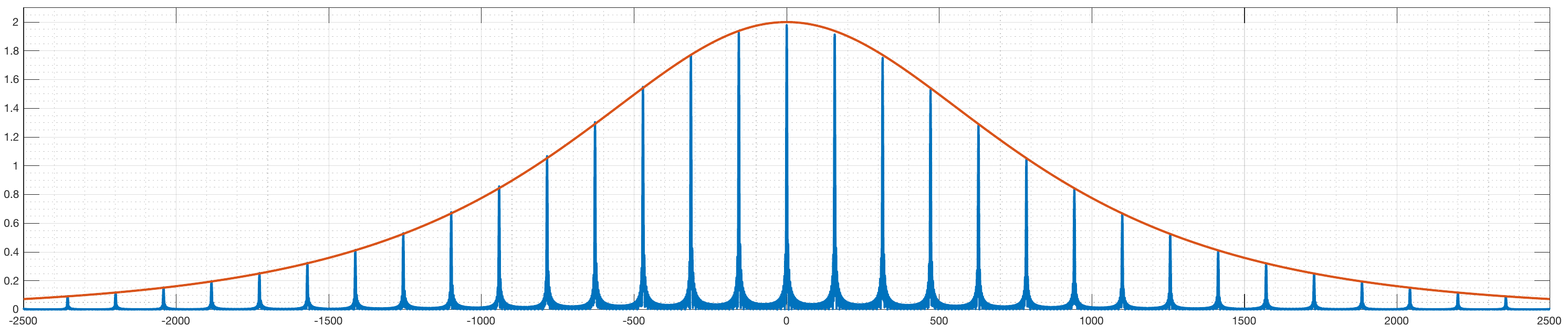}
\put(40,180){$\mathbf{a)}$}
\end{overpic}
\begin{overpic}[width=.9\linewidth]{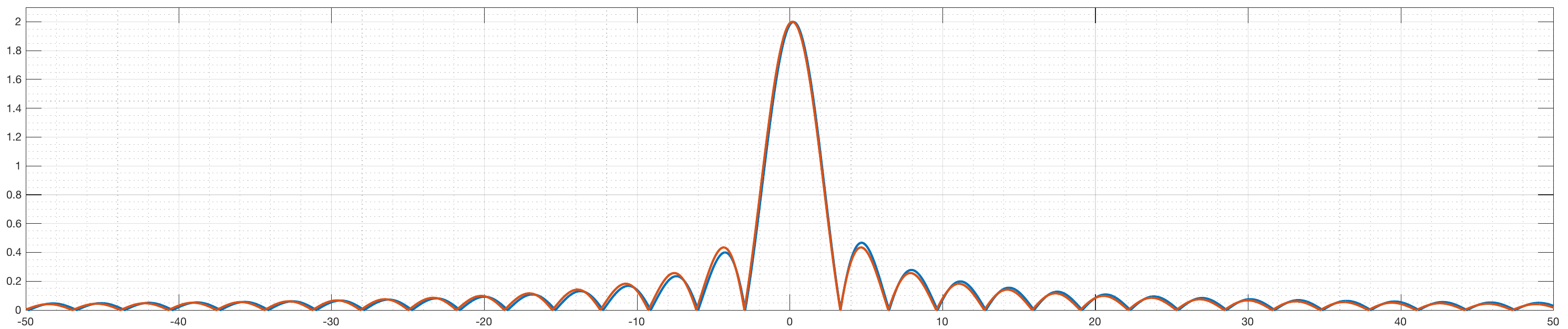}
\put(39.5,175.5){$\mathbf{b)}$}
\end{overpic}
\begin{overpic}[width=.9\linewidth]{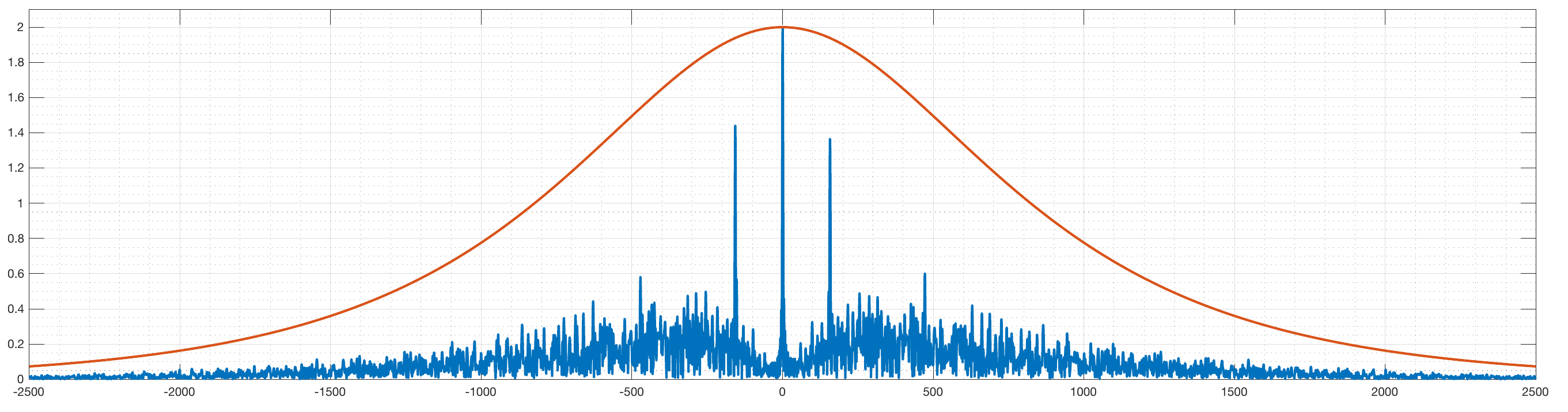}
\put(46,224.5){$\mathbf{c)}$}
\end{overpic}
\begin{overpic}[width=.9\linewidth]{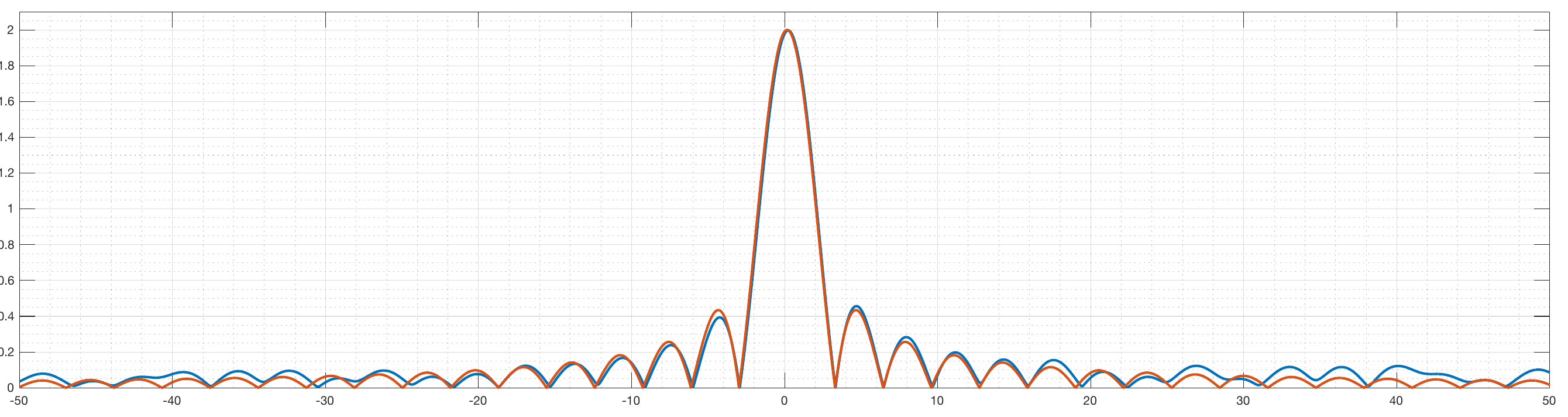}
\put(37,225.5){$\mathbf{d)}$}
\end{overpic}
\caption{
A comparison of equal spaced \textit{vs} randomly spaced velocities in Corollary~\ref{cor:sinc}.
\textbf{a)} plot of the scaled $N$-soliton solution \( \left|\frac{1}{N} \psi_N\left(\frac{2X}{N V},1\right)\right| \) at collision time (blue curve) with parameters \( N = 50 \), \( \mu = 2, V=50 \), and equally spaced velocities \( v_k = kV \), \(k=1,\dots, N\), compared with the envelope \( \mu \sech\left(\frac{2\mu X}{N V}\right) \) (red curve, see \eqref{eq:psiN.rescaled} with $T=0$)) for \( |X| < 2500 \).  
\textbf{b)} comparison between the same function (blue curve) and \( \frac{2\sin X}{X} \) (red curve) for \( |X| < 50 \).   
\textbf{c)} and \textbf{d)}: same setup as in the first two panels, but with perturbed velocities \( v_k = (k + \nu_k)V \), where \( \nu_k \) are i.i.d. random variables uniformly distributed on \( [-\tfrac{1}{5}, \tfrac{1}{5}] \).
}
\label{fig:sinc}
\end{figure}


\begin{corollary}
\label{cor:sinc}
Given constants \( \alpha\in \mathbb{R} \), and \( \mu, V >0\), consider the $N$-soliton solution $\psi_N(x,t)$ of \eqref{eq:nls} with reflectionless scattering data \eqref{eq:soliton parameters} satisfying
\begin{equation}\label{eq:scattering_data}
\mu_k=\mu, \quad v_k\in[(\alpha N+k-1)V,(\alpha N+k)V], \quad \text{and} \quad c_k = \ii \mu_k \ee^{-2\ii z_k^2},
\end{equation}
for all \( k=1,\ldots,N\), where we also assume that \( \Delta = \min\limits_{j \neq k } |v_j-v_k|  > 0 \) (necessarily \( V\geq \Delta\)). Then, there exists \( N_0 \in \mathbb N\) and \(C_*>0 \) such that whenever $ \Delta >C_*  \mu $ and \( N\geq N_0 \) it holds that
\[
\frac1N \psi_N\left( \frac{2X}{N V},1+\frac{T}{(N V)^2} \right) = \mu e^{\ii(2\alpha X-\alpha^2\frac T2)} \psi_0\left(X-\frac{\alpha T}2,T\right) + \bigo{\max\left\{\frac1N,\frac1{\Delta}\right\}},
\]
where the error is locally uniform for \( X,T\in\R \) and
\begin{equation}
\label{fresnel}
\psi_0(X,T) = - \int_0^1 \ee^{\ii\left( 2Xs - \tfrac{T}{2}s^2 \right)} \di s.
\end{equation}
In particular, when \( T=0 \), it holds that
\begin{equation}
    \frac1N \psi_N\left (\frac{2X}{N V},1\right) = -\mu \ee^{\ii (2\alpha+1)X} \frac{\sin(X)}{X} + \bigo{\max\left\{\frac1N,\frac1{\Delta}\right\}}.
    \label{sinc profile}
\end{equation}
\end{corollary}

In Corollary~\ref{cor:sinc} both \( N \) and \( \Delta \) act as large parameters and we fixed the norming constants \( c_k \) so that all the solitons collide at $(x_0,t_0)=(0,1)$. Moreover,  at time \( t=0 \), the individual solitons are well separated as their individual peaks are located at \( x_k=-v_k \), see \eqref{eq:soliton parameters}. Notice that, in the regime where $N\to \infty$, this choice of norming constants coincides with \eqref{eq:max_norming_const} from Proposition \ref{thm:darboux_intro}, since \(\displaystyle \lim_{N\to\infty}(B'(z_k))^{-1} = \ii\mu_k \). The results of Corollary~\ref{cor:sinc} are numerically illustrated in \figurename~\ref{fig:sinc} 
in the case of uniformly spaced velocities \( v_k=k V\) and with perturbed ones $v_k = k(V+\nu_k)$ where $\nu_k$ is a uniform distribution in $\left(-\frac{1}{5},\frac{1}{5}\right)$.
In fact, despite our results being valid in the large $N$ limit, \figurename~\ref{fig:sinc_5sol} illustrates a numerically good prediction of the behavior of the solution near the space-time collision point also for a small number of solitons.

\begin{figure}[t!]
\centering
\includegraphics[width=.9\linewidth]{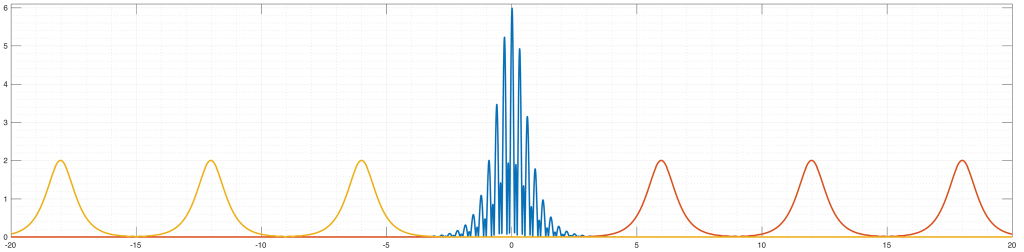}
\includegraphics[width=.9\linewidth]{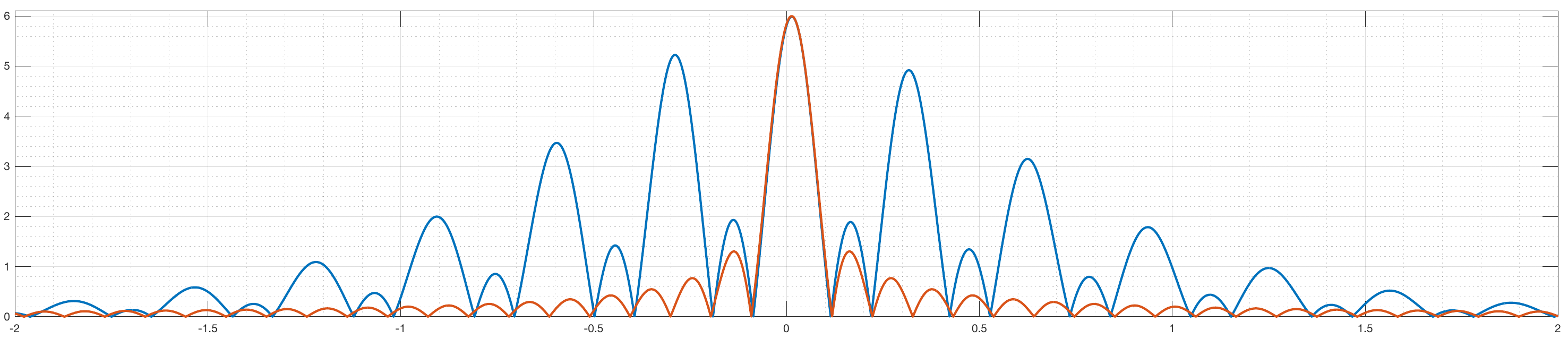}
\includegraphics[width=.9\linewidth]{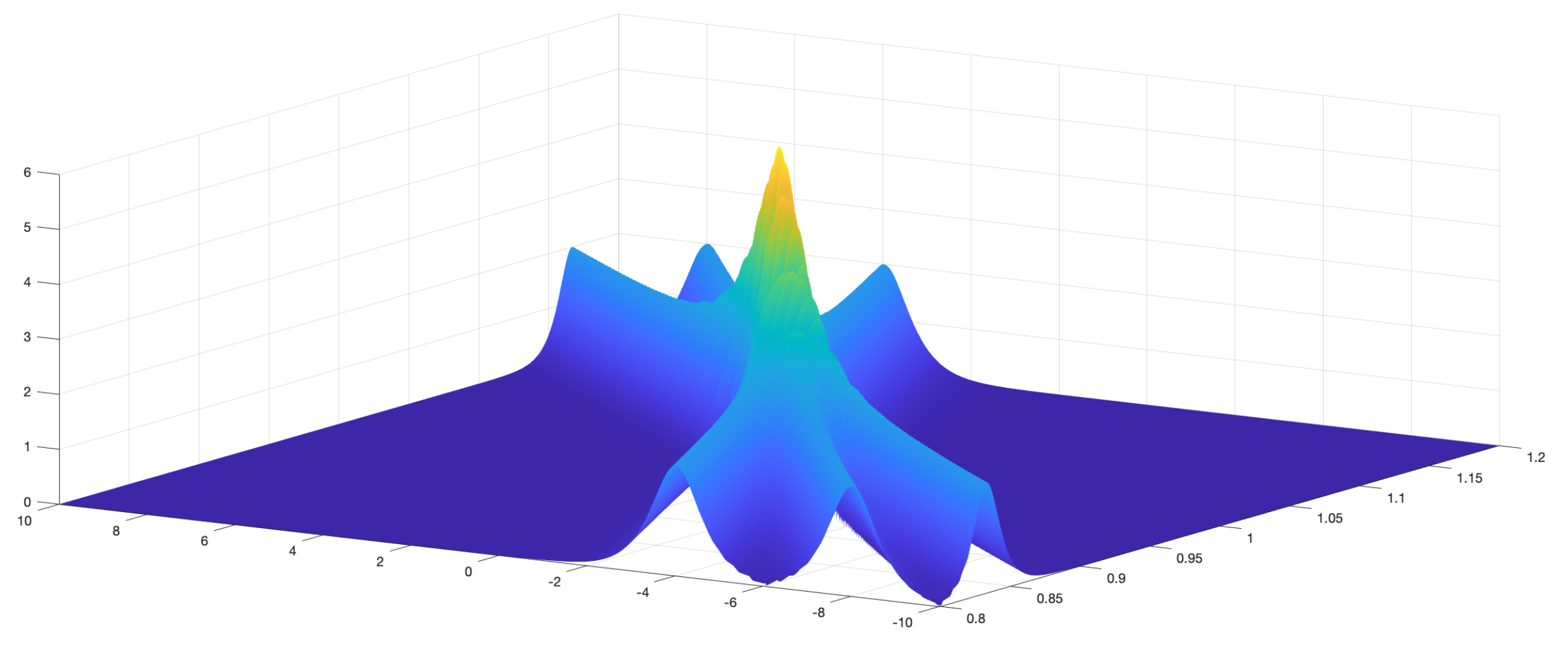}
    \caption{Top: the function \( |\psi_3(x,t)|\) with \( V=\Delta=20 \) and $\mu=2$ at times \( t=.7 \) (yellow), \( t=1 \) (blue), and \( t=1.3 \) (red) for \( |x|<20 \). Middle: comparison between \( |\psi_3(x,1)|\) (blue) and a shifted and scaled sinc profile (red) for \( |x|<2 \). Bottom: 3D graph of \( |\psi_3(x,t)|\) for \( |x|<10\) and \( |t-1|<0.2\). The setting is the one analyzed in Corollary \ref{cor:sinc} with \( v_k=k V \).}
    \label{fig:sinc_5sol}
\end{figure}

\subsection{Stochastic  Soliton Solution}

We now consider synchronized solitons with \emph{random scattering data}. We assume that  the magnitudes $\mu_k$ of the scattering data are random, while we keep the velocities $v_k$ deterministic. In order to control probabilistic fluctuations we need sufficiently strong control on the error terms in Theorem~\ref{thm:1}. To this end, we make the following assumptions.
\begin{assumption} \label{hp:random_scattering_intro}
In our probabilistic calculations we assume the following:
\begin{enumerate}
 \item The velocities are deterministic, equally spaced with spacing which grows with the number of solitons:
 \begin{equation}\label{eq:scalingDelta}
   v_k = k\Delta, \quad \Delta = \beta N^\gamma \ , \qquad \beta>0, ~~ \gamma > \tfrac{1}{2};
 \end{equation}
\item The amplitudes $\mu_k$ are independent identically distributed (i.i.d.) random variables 
\begin{equation}
\mu_k \ \sim \ \mathcal D,
\end{equation}
where $\mathcal D$ is any distribution with positive support, mean $\mu_{\mathcal D}$,  and sub exponential with parameters $(\nu,\alpha)$  (see Definition \ref{def:subexp})
\item The norming constants are random variables
\begin{equation}
    c_k = \ii \mu_k \ee^{-2\ii z_k^2}\,, \qquad z_k := \tfrac12(-k\Delta + \ii \mu_k) \ .
\end{equation}
\end{enumerate}
\end{assumption}

Under these assumptions, the $N$-soliton solution $\psi_N(x,t)$ becomes a random variable, where the solitons have random amplitudes and velocities proportional to $N$. 
In this setting, we prove a law of large numbers at the collision point showing that the emerging sinc-profile is universal, independent from the choice of the random distribution (Proposition \ref{prop:prop_CLT_intro}). This result is consistent with the deterministic case (Corollary \ref{cor:sinc}). Furthermore, we obtain central limit theorems for the fluctuations of the profile, both in the near-field region, close to the rogue wave peak (Theorem \ref{thm:local_CLT_intro}), and in the far-field region (Theorem \ref{thm:envelope_intro}). 

\begin{proposition}[\bf Universal $\sinc$-profile]
\label{prop:prop_CLT_intro}
    By choosing random soliton amplitudes according to Assumption \ref{hp:random_scattering_intro},  it holds for each fixed pair $(X,T)\in \R^2$ that
        \begin{equation}
        \frac{1}{\mu_{\mathcal D} N}\psi_N\left(\frac{2X}{\Delta N}, 1 +\frac{T}{\Delta^2N^2} \right) \to \psi_0(X,T) \quad \text{as $N\to \infty$,  in probability}\,,
    \end{equation}
    where $\mu_\cD$ is the mean value of the distribution $\cD$ and $\psi_0(X,T)$ was defined in \eqref{fresnel}.
    In particular, at collision time ($T=0$), we have 
       \begin{equation}
       \label{eq:sinc_univeral}
        \frac{1}{\mu_{\mathcal D} N} \psi_N\left(\frac{2X}{
        \Delta N},1\right) \to -\frac{\sin(X)}{X}e^{\ii X} \quad \text{as $N\to \infty$, in probability}\,.
    \end{equation}
\end{proposition}

Equation \eqref{eq:sinc_univeral} shows that a universal macroscopic profile emerges in the large $N$ limit near the collision singularity, which we illustrate numerically on \figurename~\ref{fig:sinc_universal}. 

\begin{figure}
        \centering
        \includegraphics[width=.4\textwidth]{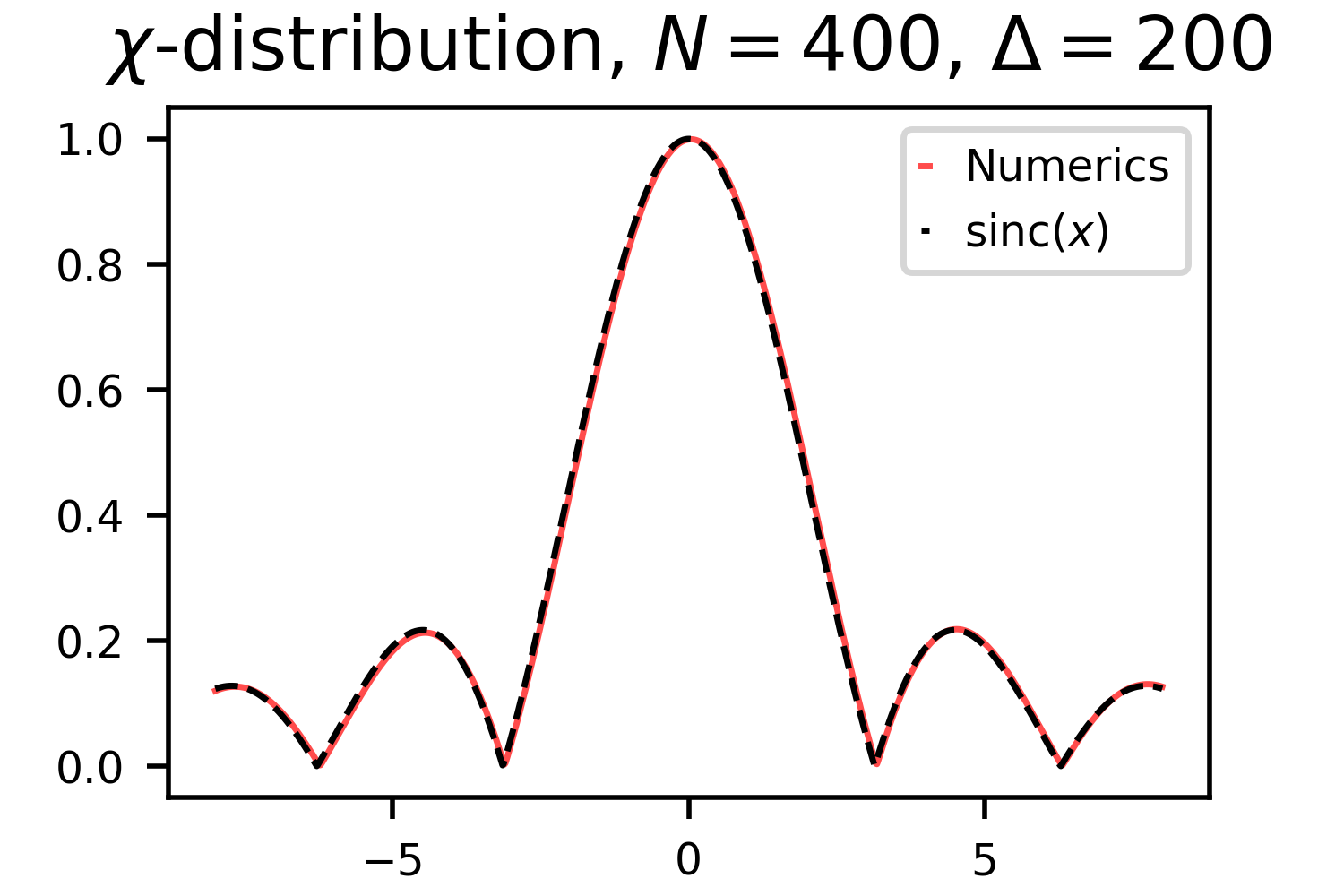}
        \includegraphics[width=.4\textwidth]{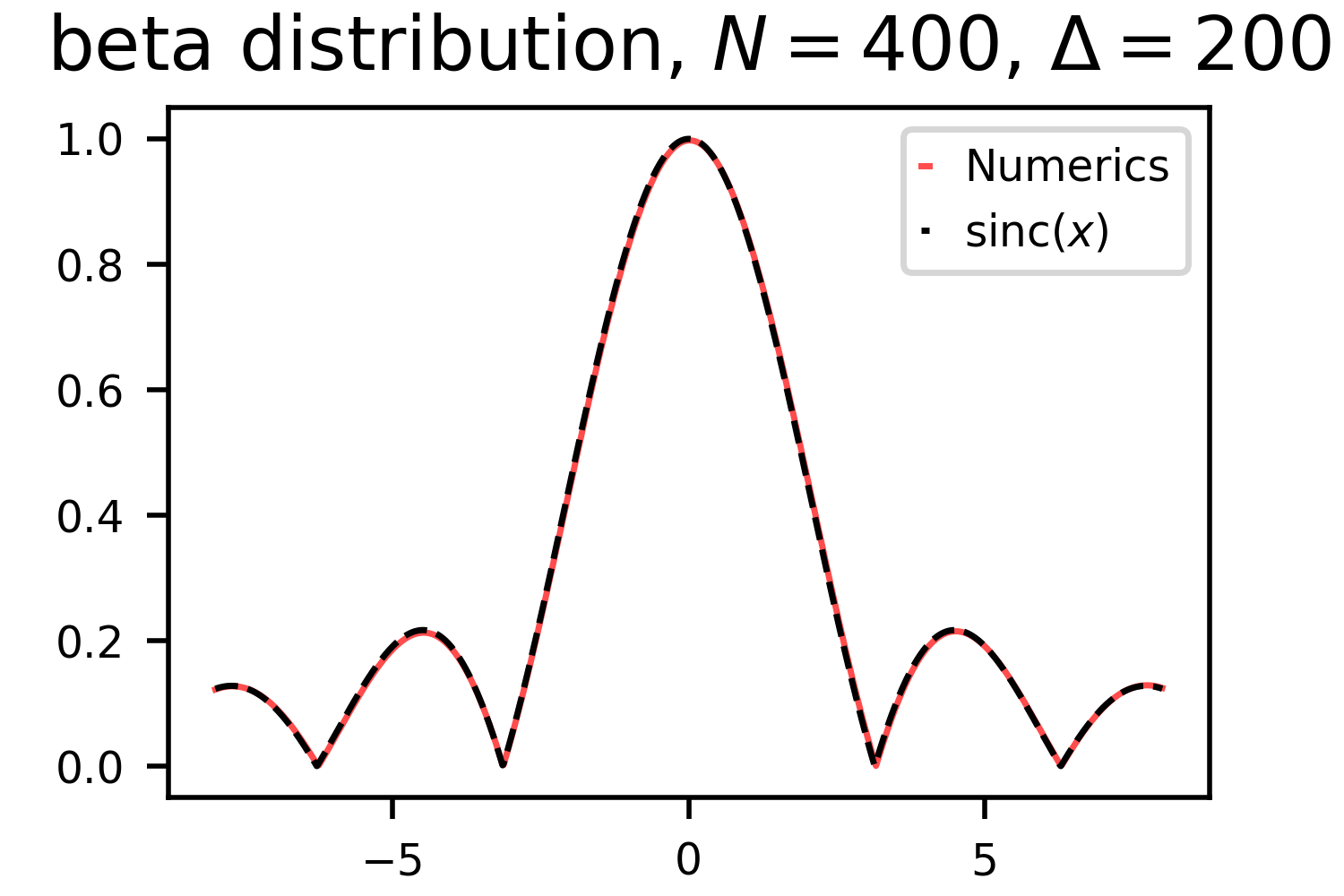}
        \includegraphics[width=.4\textwidth]{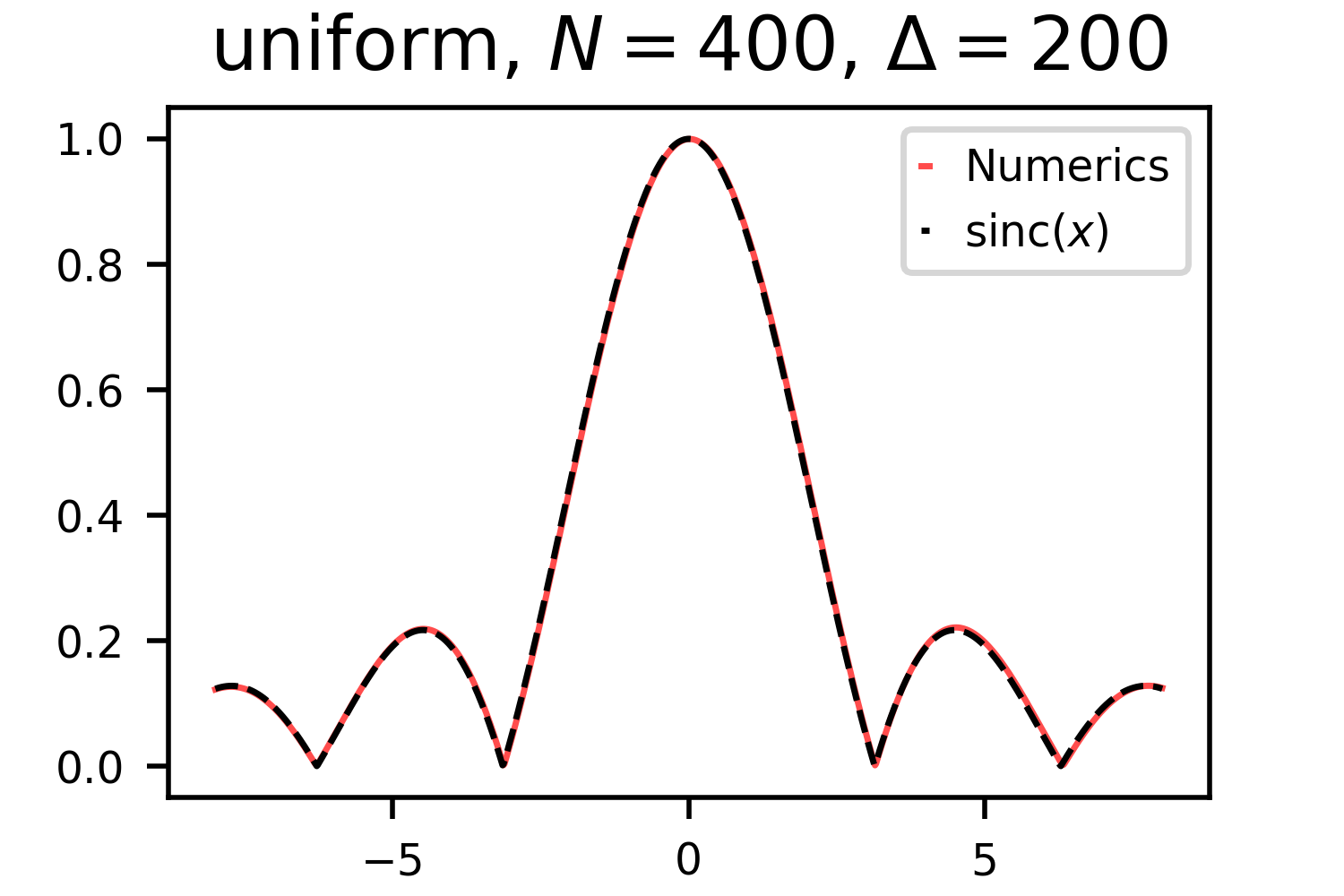}
        \includegraphics[width=.4\textwidth]{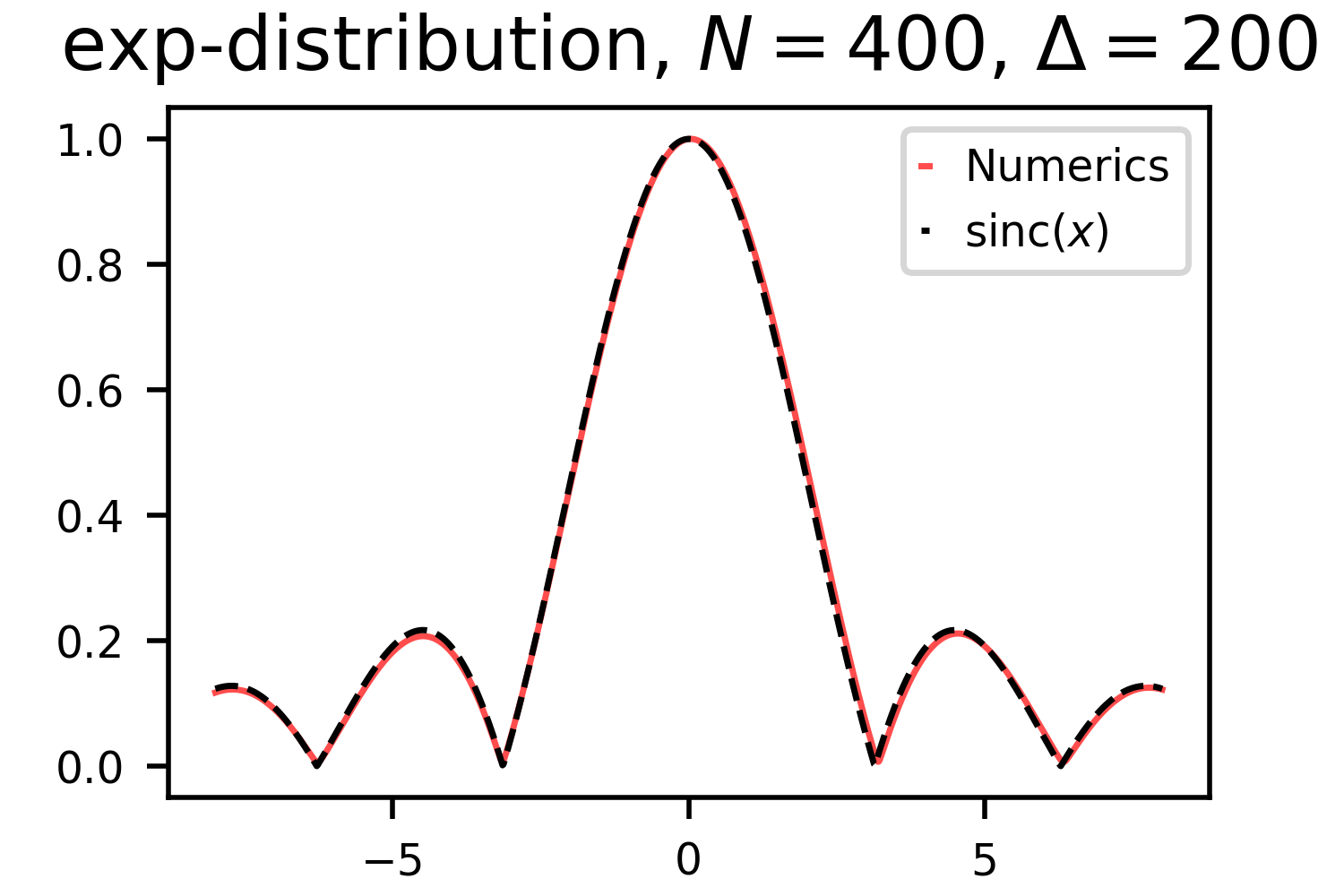}
        \caption{Numerical simulation (red) of the solution of the fNLS equation  compared to the theoretical prediction (dashed black) \eqref{eq:sinc_univeral}. The number of solitons $N=400$ and $\Delta = 200$. Here we are in the near-field, and we have both rescaled the $x$ axis, and divided the solution by  $\mu_{\mathcal{D}} N$.
        The amplitudes $\mu_j$'s are sampled according to a $\chi(2)$-distribution (top left), a $\text{Beta}_{2,2}$ distribution (top right), a uniform $(0,1)$ distribution (bottom left) and an exponential distribution with parameter $\lambda=1$ (bottom right). To realize this picture, we average over $1000$ trials.  }
        \label{fig:sinc_universal}
    \end{figure} 

In fact, Proposition~\ref{prop:prop_CLT_intro} is a consequence of the following more general  theorem.

\begin{theorem}[{\bf Central Limit Theorem in the near-field regime}]
 \label{thm:local_CLT_intro} 
    Under Assumption \ref{hp:random_scattering_intro}, for all $X, T \in \R$, the following holds
    \begin{align}
    \label{eq:clt_re}
        &  \frac{\Re\Big( \psi_{N}\lp \frac{2X}{\Delta N},1+ \frac{T}{\Delta^2N^2}\rp - N\mu_{\mathcal D}\, \psi_0(X,T)\Big)}{ \sqrt{N \Var_{\cD}}} \xrightarrow[N\to\infty]{\text{law}} \mathcal{N}\big( 0,\sigma_+(X,T)\big)\,, \\
        \label{eq:clt_im}
          &  \frac{\Im\Big(\psi_{N}\lp \frac{2X}{\Delta N},1+ \frac{T}{\Delta^2N^2}\rp - N \mu_{\mathcal D}\, \psi_0(X,T)\Big)}{ \sqrt{N \Var_{\cD}} } \xrightarrow[N\to\infty]{\text{law}}  \mathcal{N}\big( 0,\sigma_-(X,T) \big) \, ,\\
          \label{eq:clt_mod}
        &  \frac{\left \vert \psi_{N}\lp \frac{2X}{\Delta N},1+\frac{T}{\Delta^2N^2}\rp - N \mu_{\mathcal D}\,\psi_0(X,T)  \right \vert}{ \sqrt{N \Var_{\cD}}}  \xrightarrow[N\to\infty]{\text{law}}  \mathcal H(\varphi(X,T))\,,
    \end{align}
     where $\Var_{\cD} $ is the variance of the distribution $\mathcal D$,  $\mathcal H(\varphi)$ is the Hoyt distribution (see Lemma \ref{lem:CLT_approx}), and
    \[
    \sigma_{\pm}(X,T) = \frac{1}{2}\left(1 \pm \int_0^1 \cos \left( 4Xs - Ts^2 \right) \di s \right)  \ .
    \]
\end{theorem}

\begin{figure}
    \centering
    \includegraphics{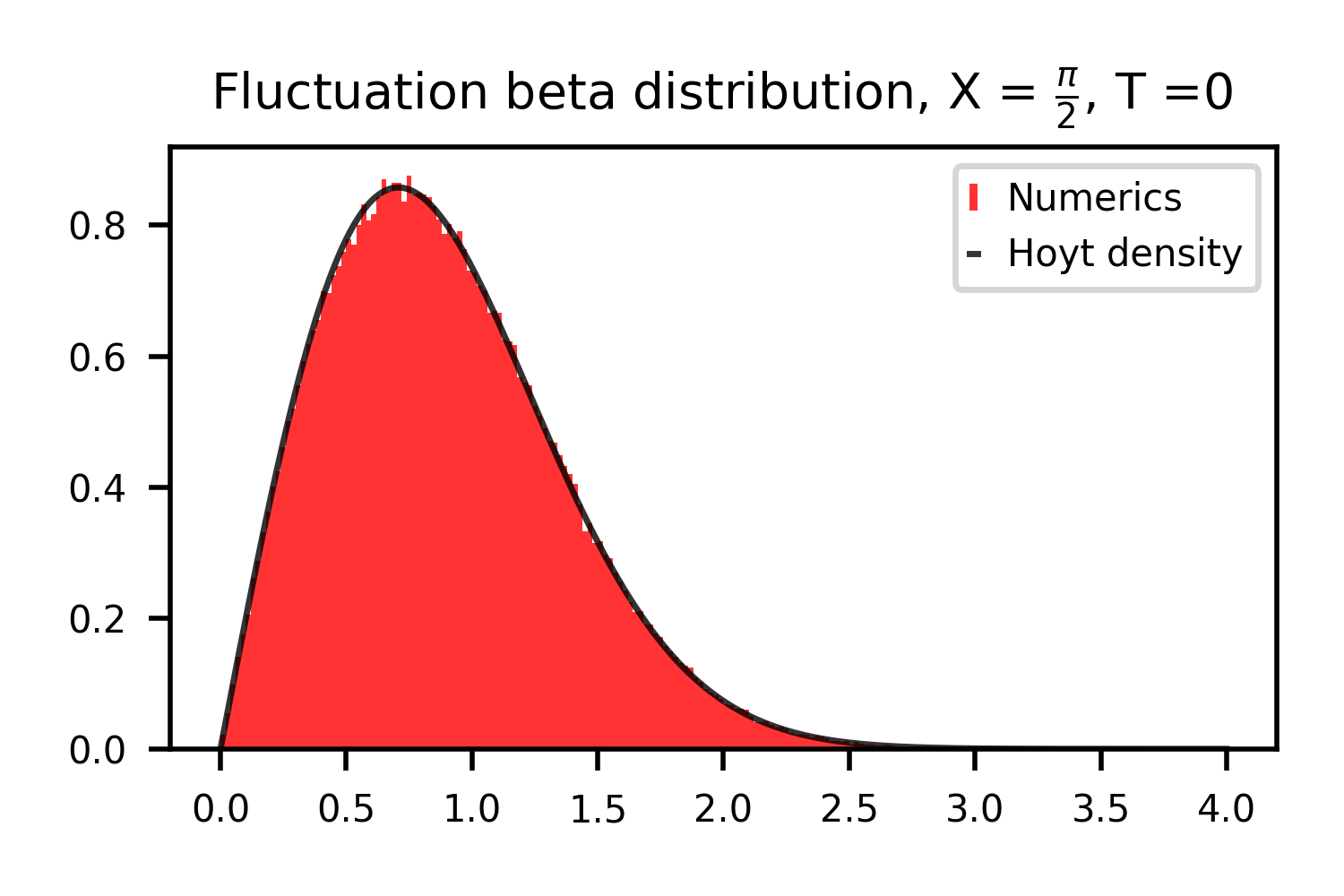}
\includegraphics{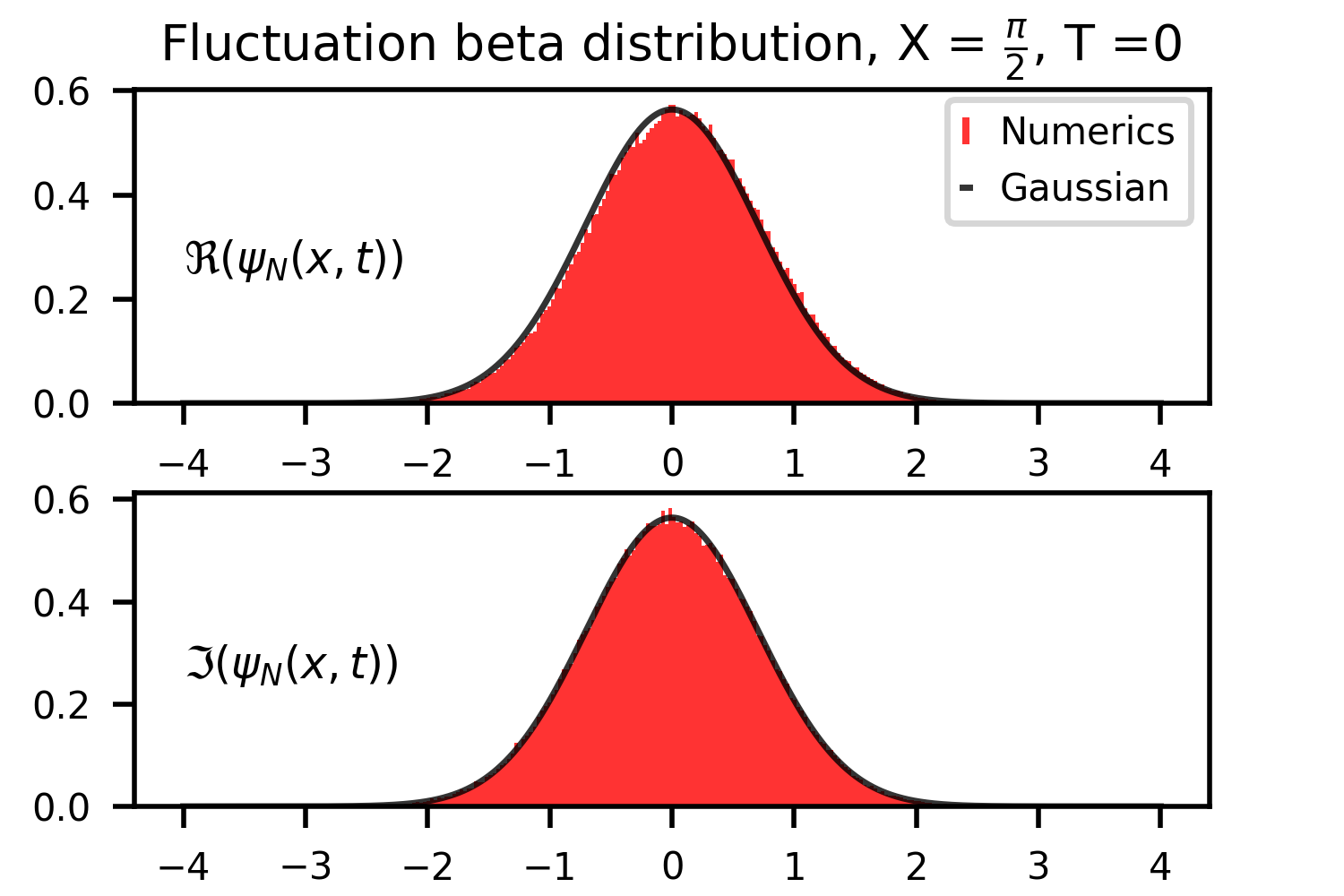}
\includegraphics{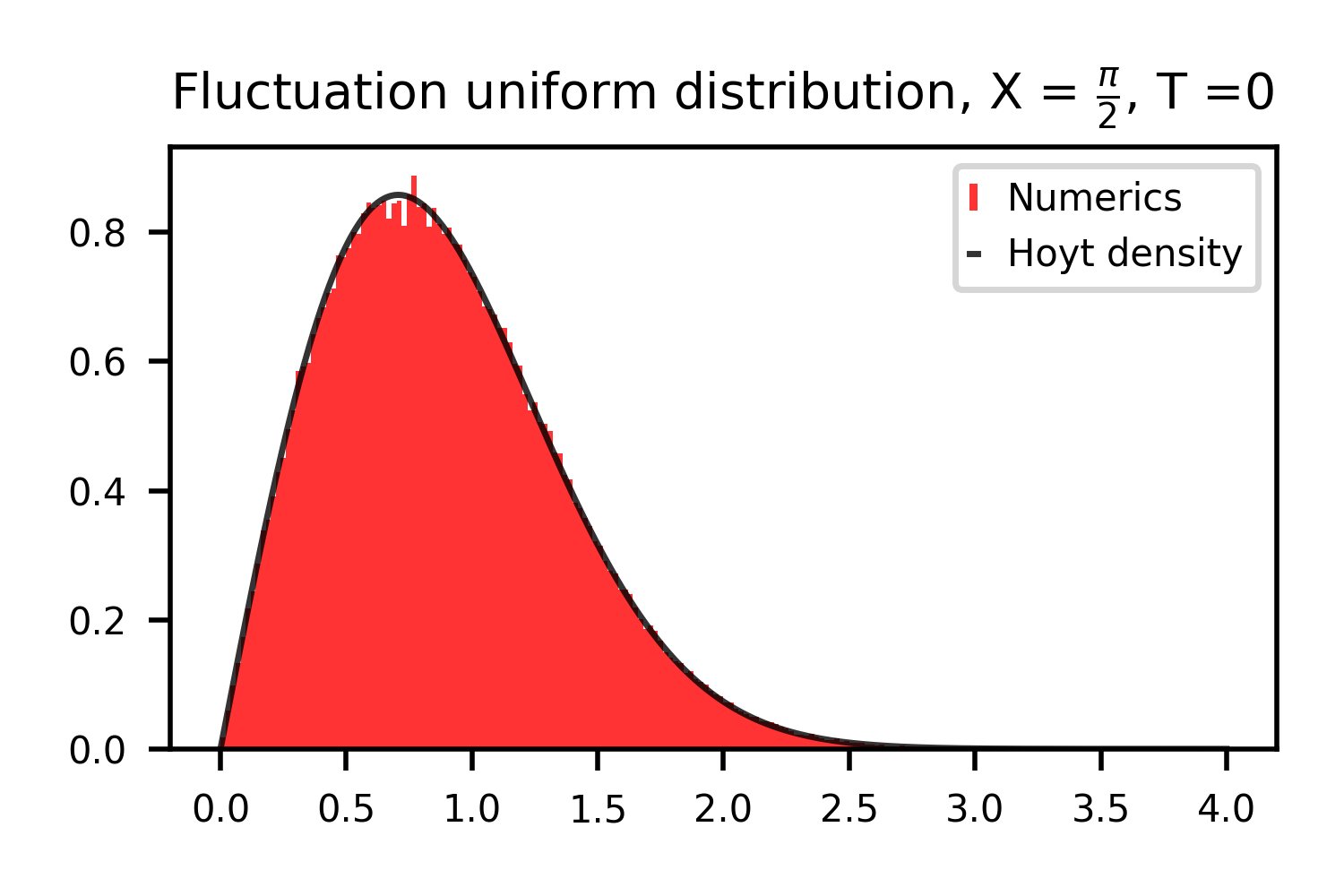}
\includegraphics{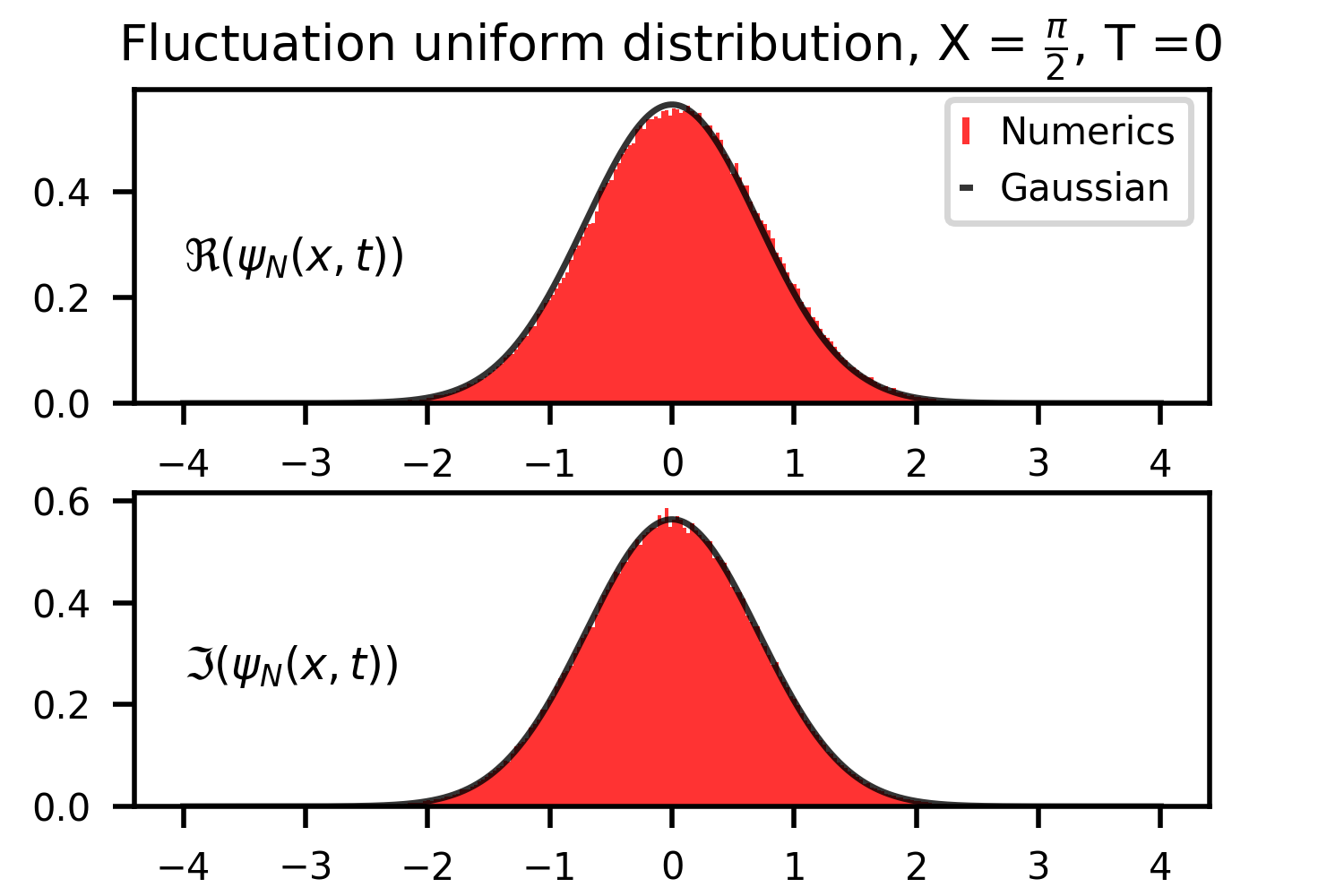}
    \caption{Fluctuation of $\psi_N(x,t)$ with respect to the average solution, $N=1200$, $\Delta = 500$, 200,000 trials. Top panel: $\mu_j$'s are i.d.d.  $\text{Beta}_{2,2}$ distribution, left side the fluctuations of $|\psi_N(x,t)|$  , right side the one of $\Re(\psi_N(x,t))$ and $\Im(\psi_N(x,t))$. Bottom panel: $\mu_j$'s are i.i.d. uniform distribution in $(0,1)$, left side the fluctuations of $|\psi_N(x,t)|$  , right side the one of $\Re(\psi_N(x,t))$ and $\Im(\psi_N(x,t))$. }
    \label{fig:fluctuations}
\end{figure}

\begin{remark}
    The variances $\sigma_\pm$ in \eqref{eq:clt_re}-\eqref{eq:clt_im} can be expressed as
\begin{multline}
    \sigma_\pm(X,T)= 
    \frac{1}{2}\left( 
    1 \pm \sqrt{\tfrac{\pi}{2T}} 
    \left\{ \cos(\xi^2) 
    \left[ 
      C\left( \sqrt{\tfrac{2}{\pi} } (\sqrt{T} - \xi) \right)
      +C\left( \sqrt{\tfrac{2}{\pi}}  \xi \right) 
    \right] 
     \right. \right. \\ \left. \left.
     + \sin(\xi^2)
      \left[
       S\left( \sqrt{\tfrac{2}{\pi} } (\sqrt{T} - \xi) \right)
       +S\left( \sqrt{\tfrac{2}{\pi}}  \xi \right) 
     \right]
    \right\}
    \right)\ ,  
\end{multline}
where $\xi(X, T) = \frac{2X}{\sqrt{T}}$, and $C(\cdot)$ and $S(\cdot)$ are the Fresnel integrals \cite[Formula 7.2.7 and 7.2.8]{DLMF}. 
\end{remark}

Notice that the limiting variances in \eqref{eq:clt_re}--\eqref{eq:clt_mod} are {independent} from the distribution $\cD$, thus universal. In \figurename~\ref{fig:fluctuations} we show the results for the Beta distribution and the uniform distribution.

\begin{remark}
\label{rem:first_clt}
For $(X,T)=(0,0)$, the variance of the normal distribution in \eqref{eq:clt_im} vanishes, which implies that the random variable  $\Im( \psi_{N}(0,0))$ is deterministically equal to $0$ in the limit as $N\to \infty$. This is consistent with the decomposition in Theorem \ref{thm:1}, where the leading term is real for $(X,T)=(0,0)$, and the remaining terms are asymptotically small with high probability (see Lemma \ref{prop:prob_error_estimate} and \eqref{eq:error_estimate_prob1}-\eqref{eq:error_estimate_prob2}).
\end{remark}

Finally, we analyze the fluctuations of the global profile of the $N$-soliton solution at collision time $\psi_N(x,1)$ over the whole spatial domain. 
\begin{theorem}[\bf Central Limit Theorem for the global profile at collision time]
\label{thm:envelope_intro}
     Let $x\in \R$. Under Assumption \ref{hp:random_scattering_intro}, consider the  \( N \)-soliton solution $\psi_{N}(x,t)$ of the fNLS equation \eqref{eq:nls}. Then the following  holds
    \begin{align}
    \label{eq:clt_envelope_re}
        &  \frac{\Re\lp\psi_N(x,1) \rp- \omega_\cD(x) \cos\lp  \frac{x\Delta(N+1)}{2}\rp D_N(x\Delta) }{\sigma_{N,\Re}(x)} \xrightarrow[N\to\infty]{\text{law}}  \cN(0,1)\,, \\
        \label{eq:clt_envelope_im}
        &  \frac{\Im\lp \psi_N(x,1) \rp - \omega_\cD(x) \sin\lp  \frac{x\Delta N}{2}\rp D_{N+1}(x\Delta)}{\sigma_{N,{\Im}}(x)} \xrightarrow[N\to\infty]{\text{law}}  \cN(0,1)\,, \quad {\text{for } x\neq 0\,,}
    \end{align}
    where $D_N(x) := \frac{\sin\lp \frac{x N}{2}\rp}{\sin \lp \frac{x}{2}\rp}$ is the Dirichlet kernel,
    \begin{align}
        & \omega_\cD(x) = - \meanval{\frac{\xi}{\cosh( x\xi)}}\, , \quad \xi\sim \cD\, , \label{eq:omega_D}\\
        \label{eq:real_var_envelope}
        &\sigma^2_{N,\Re}(x) = \Var\lp\frac{\xi}{\cosh( x\xi)}  \rp\lp \frac{N-1}{2} + \frac{1}{2} \cos\lp x\Delta N \rp D_{N+1}(2x\Delta) \rp \,,  \\ 
        &\sigma^2_{N,\Im}(x) = \Var\lp\frac{\xi}{\cosh( x\xi)}  \rp\lp \frac{N+1}{2} - \frac{1}{2}  \cos\lp x\Delta N \rp D_{N+1}(2x\Delta) \rp \, ,
    \end{align}
    and $\Var(\cdot) $ is the variance of the given random variable. Moreover, 
    \begin{equation}
    \label{eq:clt_envelope_mod}
        \lim_{N\to\infty} \frac{1}{N} \Big( \left|\psi_N(x,1)\right| - \left|\omega_\cD(x)D_N(x\Delta)\right|\Big) \to 0
    \end{equation}
    in probability.

\end{theorem}

We define the function $|\omega_\cD(x)|$ as the {\em envelope}, meaning a smooth curve outlining the extremes of an oscillating signal. 
As an example, if $\cD \sim \chi^2(2)$, 
the envelope is equal to
\[
 |\omega_{\mathcal D}(x)| = \frac{\Psi ^{(1)}\left(\frac{| x| +1}{4 | x| }\right)-\Psi ^{(1)}\left(\frac{1}{4} \left(3+\frac{1}{| x| }\right)\right)}{8 x^2}\,,
\]
where $\Psi^{(1)}$ is the $1^{st}$ poly-gamma function \cite[Ch. 5]{DLMF}. 
In \figurename~\ref{fig:envelope}, we show the envelope profile $|\omega_{\mathcal D}(x)|$, its modulation through the Dirichlet kernel $ D_N(x\Delta)$ and the numerical average of the solution for several choice of distributions $\cD$. 

\begin{figure}[t!]
    \centering
    \includegraphics{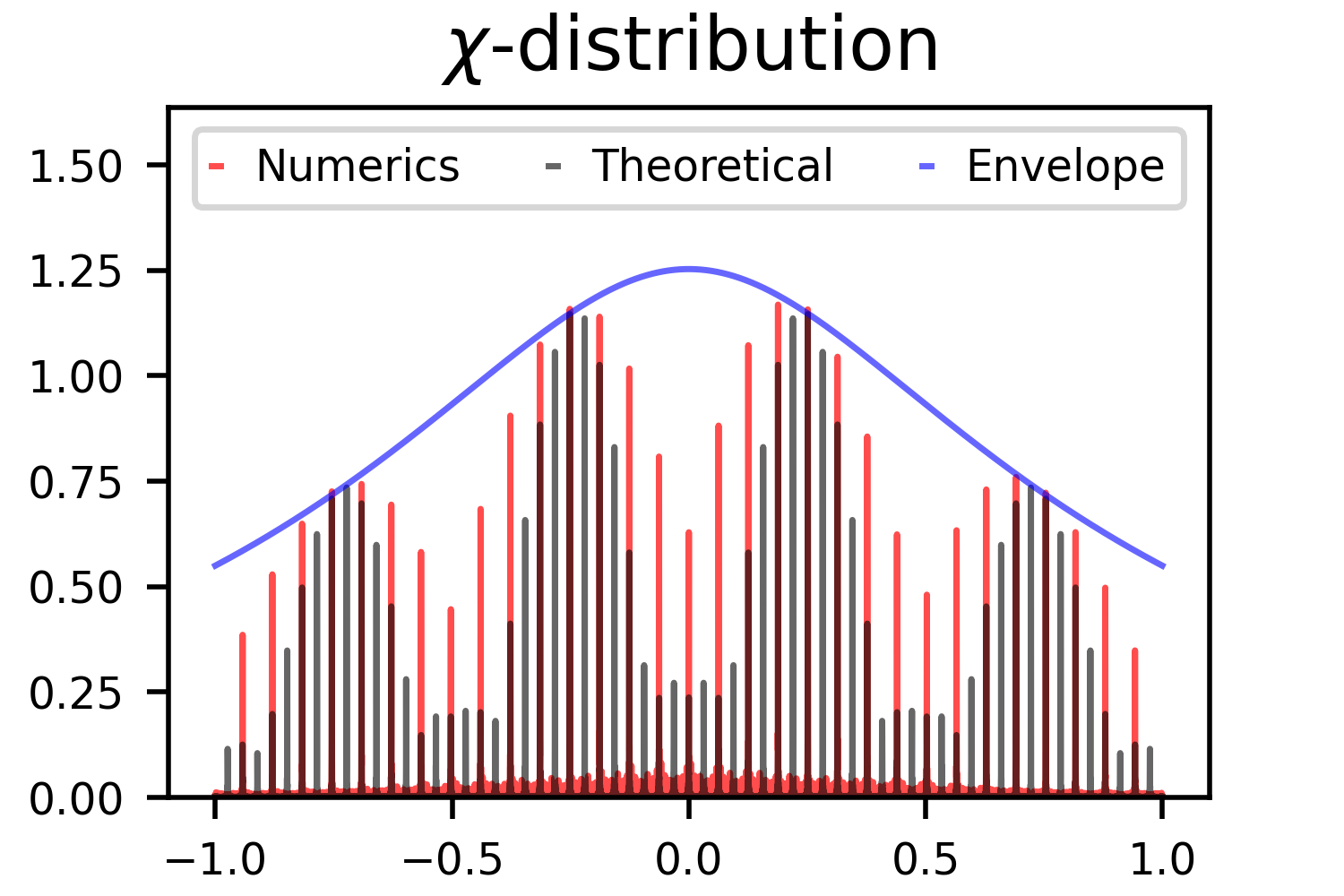}
    \includegraphics{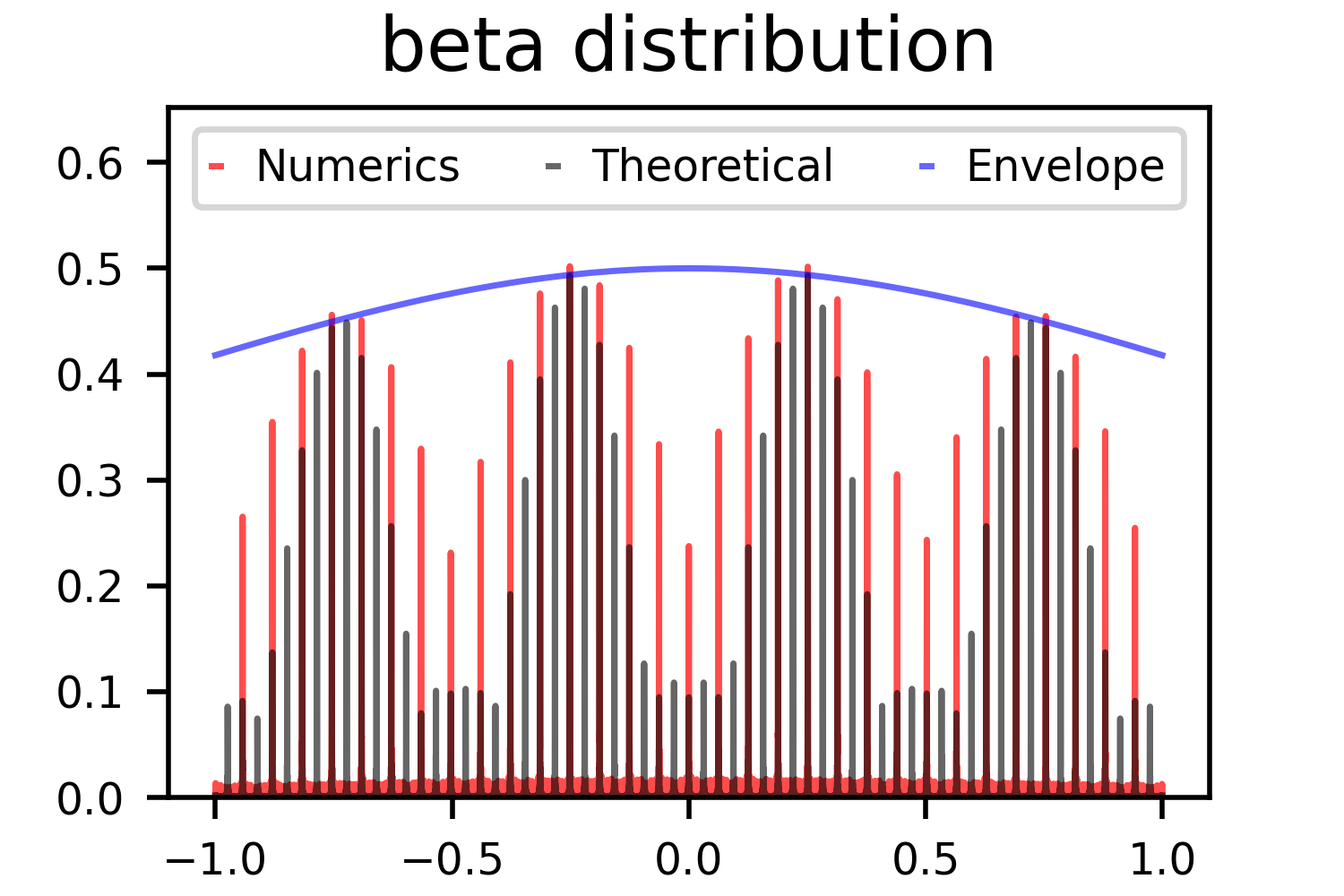}
    \includegraphics{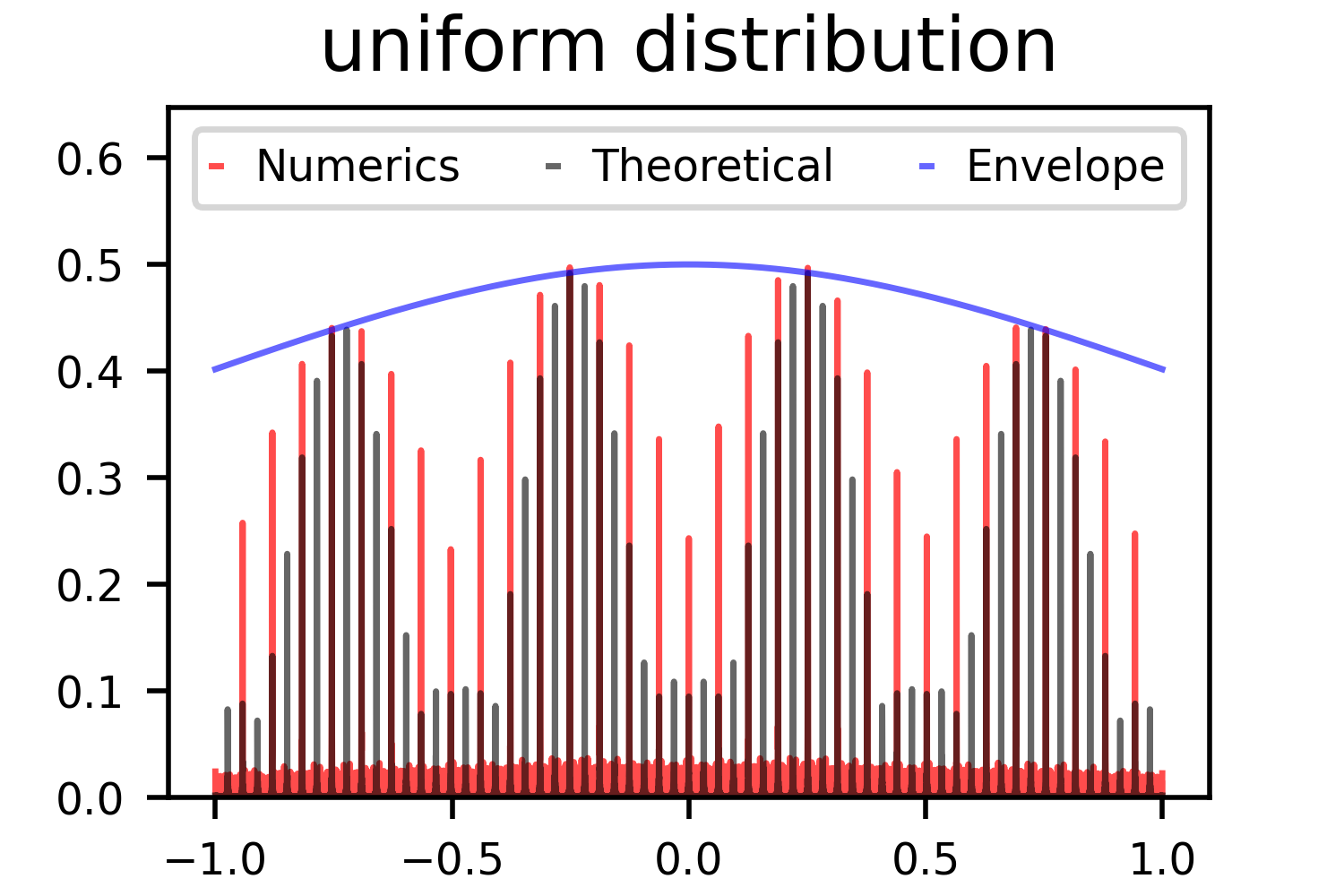}
    \includegraphics{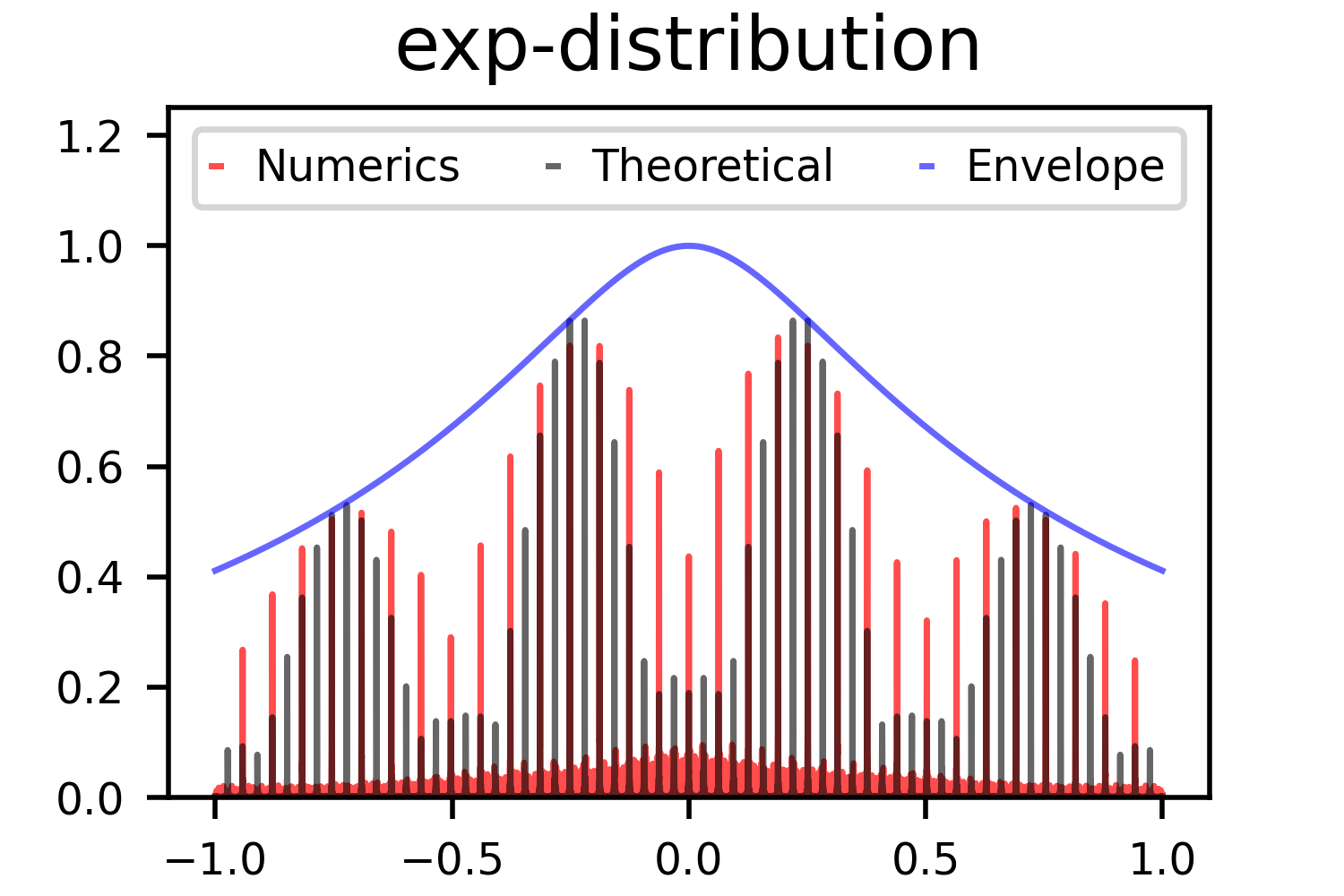}    
    \caption{Macroscopic profile: numerical simulation (red) of $|\psi_N(x,1)|$, theoretical prediction (dashed black) \eqref{eq:clt_envelope_mod} and the envelope $|\omega_{\mathcal D}(x)|$ (blue) \eqref{eq:omega_D}. The number of solitons $N=400$ and $\Delta = 200$. The $\mu_j$'s are sampled according to a $\chi(2)$-distribution (top left), a $\text{Beta}_{2,2}$ distribution (top right), a uniform $(0,1)$ distribution (bottom left) and an exponential distribution with parameter $\lambda=1$ (bottom right). To realize this picture, we averaged over $1000$ trials. }
    \label{fig:envelope}
\end{figure}

\medskip

In recent years, a lot of effort has been put into describing the formation of rogue waves in deep sea and optical fibers.  Informally, a rogue wave is a large-amplitude disturbance of the background state. Historically, the first instance of a rogue wave solution was derived by Peregrine \cite{Peregrine83}. Through the years, many more solutions with similar behavioural pattern have been studied experimentally, numerically and analytically (see for example \cite{Chabchoub,Tikan2017,Dematteis2018,Chabchoubetal11}).
Furthermore, numerical and physical experiments have observed  that the interaction of suitably prepared solitons also yield rogue waves (see \cite{MussotKudlinski09, KharifPelinovsky01, Sun16} for the NLS and \cite{SlunyaevPelinovsky16} for the modified KdV equations).
 
In the experiments conducted in \cite{Chabchoub} the authors studied rogue wave formation in a deep water tank: using the NLS equation as a model, they argue that such events are usually caused by the presence of a Peregrine breather appearing in the dynamics, or a degenerate two-soliton solution.
In \cite{Tikan2017}, it has been shown that the Peregrine breather emerges as a universal profile as the compression of the $N$-soliton solution to the NLS equation; furthermore, at the point of maximal localization, it yields to a peak three times bigger than the background.

Certain types of rogue waves have been extensively studied in \cite{BilmanBuckingham19, BilmanLingMiller20, BilmanMiller22, BilmanMiller24a, BilmanMiller24b} via a careful Riemann--Hilbert analysis: the authors showed that rogues waves can be constructed from high-order breather solutions \cite{BilmanLingMiller20, BilmanMiller22}, high-order soliton solutions \cite{BilmanBuckingham19}, or high-order solutions belonging to a one-parameter family which encompasses the previous two classes \cite{BilmanMiller24a}. In the limit as the order goes to infinity, the solution displays a universal central peak, which is described in terms of a member of the Painlev\'e III hierarchy \cite{Sakka}.

 Comparing with the previous literature, the results presented in this paper prove rigorously that the formation of rogue waves can be the result of the constructive interaction of a handful of solitons at one point in space-time (see \figurename~\ref{fig:sinc_5sol}), and the resulting peak is universal, as it survives random perturbations of the soliton amplitudes.
 
 Our soliton solution setup is similar to the scenario described in \cite{SlunyaevPelinovsky16}: the authors consider a multi-soliton solution of the modified KdV equation, and tune the scattering data to obtain a rogue wave at a given collision time, so that the height of the peak is equal to the sum of amplitudes. 
 Furthermore, phenomena of coherent soliton pulse trains, modelled by a generalization of the fNLS equation, appear in experimental optics, specifically in (micro)-resonators and dissipative Kerr soliton combs \cite{Kippenberg14, Kippenberg18}.

We finally highlight that the $N$-soliton configuration studied in this paper is an instance of a {\em dilute soliton gas} \cite{Suret2023SolitonGT}. This model was originally proposed by Zakharov \cite{Zakharov71} as an infinite collection of KdV solitons with random parameters, which are so sparse that it is possible to distinguish individual soliton-soliton interactions. Zakharov additionally derived a formula to describe the average velocity of a trial soliton as it travels through the diluted KdV gas. A more general kinetic formula for the soliton gas solving the fNLS equation was later derived in \cite{ElTovbis20}, however we  do not pursue this direction of research here.

\section{Deterministic $N$-Soliton Solutions}\label{sec:RHP_Nsolitons}

The proofs of Proposition~\ref{thm:darboux_intro}, Theorem~\ref{thm:1}, and Corollary~\ref{cor:sinc}, which we prove in this section, rely on the integrability of the fNLS equation.  

\subsection{Darboux (Dressing) Method}

The integrability of fNLS was established by Zakharov and Shabat in \cite{zakharov1972exact} where they showed that \eqref{eq:nls} has a Lax Pair structure given by 
\begin{subequations}\label{ZH system}
\begin{align}
\label{ZH system_a}
\bbs{\Phi}_x &= \mathcal{L} \bbs{\Phi}, \quad \mathcal{L}:=\begin{bmatrix}\ -\ii z & \psi \\ -\overline{\psi} & \ii z \end{bmatrix},
\\
\label{ZH system_b}
\bbs{\Phi}_t &= \mathcal{B} \bbs{\Phi}, \quad \mathcal{B}:=\begin{bmatrix}\ -\ii z^2 + \frac{\ii}{2} |\psi|^2 & z \psi + \frac{\ii}{2} \psi_x \\ -z \overline{\psi} + \frac{\ii}{2} \overline{\psi}_x & \ii z^2 - \frac{\ii}{2} |\psi |^2 \end{bmatrix} ,
\end{align}
\end{subequations}
and established the existence of a simultaneous solution of this overdetermined system of ODEs provided the Lax operators $\mathcal{L}$ and $\mathcal{B}$ satisfy the compatibility condition $\mathcal{L}_t - \mathcal{B}_x + \mathcal{L} \mathcal{B} - \mathcal{B} \mathcal{L} = 0$, which is equivalent to $\psi$ being a solution of \eqref{eq:nls}. 

The integrable structure allows us to compute solutions via the Inverse Scattering Transform (IST) method \cite{IST, novikov1984theory,faddeev2007hamiltonian, zakharov1972exact}. 
The formulation of the IST method starts by considering the scattering problem for the first operator \eqref{ZH system_a} in the fNLS Lax Pair viewed as an eigenvalue problem:
\begin{equation}
\widehat{\mathcal{L}}\bbs{\Phi} = z \bbs{\Phi},
\quad
\widehat{\mathcal{L}} = \ii \begin{bmatrix}\ 1 & 0 \\ 0 & -1 \end{bmatrix}\frac{\pt}{\pt x} - \ii\begin{bmatrix}\ 0 & \psi \\ \overline{\psi} & 0 \end{bmatrix}\ . \label{ZH system1}
\end{equation}
For spatially localized potentials $\psi$, the spectrum of \eqref{ZH system1} generically\footnote{By "generic", we refer to potentials belonging to an open dense subset of $L^1(\R)$ (see \cite{BealsCoifmann1984}).} consist of a finite number of non-real points (discrete spectrum) and the real line (continuous spectrum). In this setting the scattering data, which is time dependent, consists of a \emph{reflection coefficient} $r: \R \to \C$ defined on the continuous spectrum, the collection of the discrete eigenvalues $z_{k} \in \C^+$ (the \emph{poles}), and the so-called \emph{norming constants} $C_{k} \in \C^*$ associated to each discrete eigenvalue in the following sense: for each $z_k \in \C^+$ there exist vector solutions of \eqref{ZH system1} 
\[
    \bbs{\phi}^+(x,t;z_k) \sim \begin{bmatrix} 0 \\ 1 \end{bmatrix} e^{\ii z_k x}, \quad x \to +\infty,  \qquad
    \bbs{\phi}^-(x,t;z_k) \sim \begin{bmatrix} 1 \\ 0 \end{bmatrix} e^{-\ii z_k x}, \quad x \to -\infty,
\]
such that \( \bbs{\phi}^+(x,t;z_k) = C_k(t)\bbs{\phi}^-(x,t;z_k) \). The key result, which makes the IST effective, is that the time evolution of the scattering data is trivial, i.e. \( \mathcal{S}(t) = ( \{z_k(t), C_k(t)\}_{k=1}^N, r(z;t)) \), the scattering data at time \( t \), is given by
\begin{equation}
\label{evolution}
    \mathcal{S}(t) = \Big(  \left\{ \left( z_k , \ C_ke^{-2\ii z_k^2 t} \right) \right\}_{k=1}^N, ~ r(z)e^{2\ii z^2 t} \Big) \ ,    
\end{equation}
where \( \mathcal{S}(0) = ( \{(z_k, C_k\}_{k=1}^N, r(z)) \) corresponds to the initial data $\psi_0(x) = \psi(x,0)$. The IST is the process by which one recovers the time-evolved potential $\psi(x,t)$ from the known evolution of the scattering data $\mathcal{S}(t)$. There are several ways to formulate the inverse scattering transform depending on the complexity of the scattering data. In what follows, we are only interested in the reflectionless case \( r=0 \). The corresponding solution is the $N$-soliton solution $\psi_N(x,t)$, and it can be obtained iteratively via the Darboux transform (or dressing method) \cite{Gelash2021,Gelash2020}, which we recall here briefly. As usual, we write \( z_k = (-v_k+\ii\mu_k)/2\). Let \( \bbs{\Phi}_n(z;x,t) \), \( n\geq 1 \), be the solution of the ZS system~\eqref{ZH system_a}-\eqref{ZH system_b}. These matrices can be constructed inductively using the so-called dressing matrices $\boldsymbol{\chi}_n(z;x,t)$ via
\begin{align}
\bbs{\Phi}_n(z;x,t) &= \boldsymbol{\chi}_n(z;x,t)\bbs{\Phi}_{n-1}(z;x,t), \qquad \bbs{\Phi}_{0}(z;x,t) = \ee^{-\ii xz \boldsymbol \sigma_3}, \label{dressing Psi}\\
(\boldsymbol{\chi}_n(z;x,t))_{{j\ell}} &= \delta_{{j\ell}} + \frac{z_n - \overline z_n}{z - z_n} \frac{\overline{q_{n{j}}(x,t)} \, q_{n{\ell}}(x,t)}{|\bbs{q}_n(x,t)|^2}, \label{dressing matrix} \\
\bbs{q}_n(x,t) &= \overline{\bbs{\Phi}_{n-1}(\overline z_n;x,t)} \, \begin{bmatrix} 1 \\ C_n(t) \end{bmatrix}, \label{qn via Phin}
\end{align}
where ${j,\ell}=1,2$ and $\delta_{{j\ell}}$ is the Kronecker delta in \eqref{dressing matrix}, and we write \( \bbs{q}_n(x,t) = \begin{bmatrix}q_{n1}(x,t) &q_{n2}(x,t)\end{bmatrix}^{\top} \). The  $N$-soliton solution \( \psi_N(x,t)\) is then given by
\begin{equation}
\label{psi_n}
\psi_N(x,t) = -2\sum_{k=1}^N\mu_k\frac{\overline{q_{k1}(x,t)}\, q_{k2}(x,t)}{|\bbs{q}_k(x,t)|^2}.
\end{equation}
Because \( 2\left|\overline{q_{k1}}\, q_{k2}\right| \leq |\bbs{q}_k|^2 \) for all $(x,t)\in \R\times \R^+$, \eqref{psi_n} readily yields the bound in \eqref{eq:upper_bound_result}.

\subsection{Riemann-Hilbert Approach}
The dressing method is a straightforward method for iteratively computing soliton solutions or, more generally, for adding solitons to a known ``seed'' solution. However, it is not well suited to asymptotic analysis. 
To study the dependence of solutions on external parameters, the IST is more appropriately formulated as a Riemann-Hilbert problem (RHP). Using the nonlinear steepest descent method one can often compute complete asymptotic expansions to the solution of the RHP with explicit error bounds. 
Below, the Riemann-Hilbert problem for fNLS is given for reflectionless potentials, which is all we need for our purposes. The RHP for generic potentials that decay sufficiently as $|x| \to \infty$ can be found many places. See, for example, \cite{Kamvissis2003, Jenkins2018}. 

\begin{rhp}
\label{rhp-sch}
    Given reflectionless scattering data $\left\{ \left(z_k, C_k\right) \right\}_{k=1}^N$, find a matrix function $\bbs M(z)\in\textrm{SL}(2,\C)$ such that
\begin{enumerate}
	\item $\bbs{M}(\,\cdot\;; x,t)$ is analytic in $ \C \setminus \big( \{ z_k, \overline z_k \}_{k=1}^N\big)$;
	\item $\bbs{M}(z;x,t) = \bbs I + \bigo{z^{-1}}$ as $z \to \infty$;
	\item $\bbs{M}(z;x,t)$ has a simple pole at each $z_k$ and \( \overline z_k \) satisfying the residue relation
	\begin{equation}
		\begin{aligned}
		&\Res_{z=z_k} \bbs{M}(z;x,t) = \phantom{-}c_k \ee^{2 \ii \theta(z_k;x,t)} \lim_{z \to z_k} \bbs{M}(z;x,t) \tril[0]{1}, \\
		&\Res_{z=\bar{z}_k} \bbs{M}(z;x,t) = -\overline c_k \ee^{-2\ii \theta(\overline z_k;x,t)} \lim_{z \to \bar{z}_k} \bbs{M}(z;x,t) \triu[0]{1},
		\end{aligned}
	\end{equation}
 where \( \theta(z;x,t) = tz^2 +xz \) and the {\em residue coefficients} $c_k$ are related to the norming constants $C_k$ by (see \cite{Gelash2020,Gelash2021}) 
\begin{equation}
\label{eq:c_C_rel}
    c_k := \frac{1}{C_k\, B'(z_k)}\,, \quad B(z) = \prod_{k=1}^N\frac{z-z_k}{z-\overline z_k}.
\end{equation}
\end{enumerate}
\end{rhp}

Expanding the solution of this RHP as $z \to \infty$, it can be shown \cite{Jenkins2018} that
\begin{equation}\label{m.expand}
    \bbs{M}(z;x,t) = \bbs{I} + \frac{1}{2\ii z} 
    \begin{bmatrix}
        -m(x,t) & \psi(x,t) \smallskip \\ \overline{\psi(x,t)} & m(x,t) 
    \end{bmatrix}
    +\bigo{z^{-2}}, \quad z \to \infty,
\end{equation}
where 
\begin{equation} \label{eq:moment}
m(x,t) := \int_x^\infty | \psi(s,t)|^2 \, \di s .
\end{equation}
It follows that the solution of the fNLS equation \eqref{eq:nls} can be recovered as 
 \begin{equation}\label{eq:recovery}
     \psi(x,t) = \lim_{z\to\infty} 2\ii z [\bbs M(z;x,t)]_{1,2}\,.
 \end{equation}

\subsection{Proof of Proposition \ref{thm:darboux_intro}}
\label{subsec:Darboux}

According to \eqref{psi_n}, to prove \eqref{upper bound}, it is enough to show that 
\begin{equation}
\label{inductive hyp}
q_{k1}(x_0,t_0)=q_{k2}(x_0,t_0).
\end{equation}
Since \( C_k(t_0) = C_k\ee^{-2\ii z_k^2t_0} \), our choice of the constants \( c_k \) and  \eqref{eq:c_C_rel} yields that \( C_k(t_0) = \ee^{2\ii z_kx_0} \) for each \( k=1,\ldots,N \). Hence, we get from \eqref{qn via Phin} and \eqref{dressing Psi} that
\[
\bbs{q}_1(x_0,t_0) = \ee^{\ii x_0z_1} \begin{bmatrix} 1 \\ 1 \end{bmatrix} \quad \Rightarrow \quad q_{11}(x_0,t_0)=q_{12}(x_0,t_0)
\]
as desired. Assume now that \eqref{inductive hyp} holds for all \( k=1,\ldots,n-1\). Relations \eqref{qn via Phin} and \eqref{dressing Psi} give
\[
\bbs q_n(x_0,t_0) = \ee^{\ii x_0z_n}\left(\prod_{k=1}^{n-1}\overline{\boldsymbol{\chi}_k(\overline z_n;x_0,t_0)}\right) \begin{bmatrix} 1 \\ 1 \end{bmatrix}.
\]
It follows from \eqref{inductive hyp} and \eqref{dressing matrix} that
\[
\boldsymbol{\chi}_k(z;x_0,t_0)\begin{bmatrix} 1 \\ 1 \end{bmatrix} = \left( \boldsymbol I + \frac12 \frac{z_k-\overline z_k}{z-z_k} \begin{bmatrix} 1 & 1 \\ 1 & 1 \end{bmatrix} \right)\begin{bmatrix} 1 \\ 1 \end{bmatrix} = \left( 1 + \frac{z_k-\overline z_k}{z-z_k} \right) \begin{bmatrix} 1 \\ 1 \end{bmatrix},
\]
which, in turn, readily implies that
\[
\bbs q_n(x_0,t_0) = \ee^{\ii x_0z_n}\prod_{k=1}^{n-1} \left( 1 - \frac{z_k - \overline z_k}{z_n-\overline z_k} \right) \begin{bmatrix} 1 \\ 1 \end{bmatrix}  \quad \Rightarrow \quad q_{n1}(x_0,t_0)=q_{n2}(x_0,t_0)
\]
as desired. This clearly finishes the proof of the inductive step and therefore of \eqref{upper bound}.

\subsection{Proof of Theorem~\ref{thm:1}}
\label{subsec:RHP_Nsols_train}

Given the RHP~\ref{rhp-sch}, we first perform a rescaling transformation
\[
	\widehat{\bbs{M}}(\lambda;x,t) := \bbs{M}(\Delta \lambda; x,t) \ .
\]
Intuitively, since $\Delta \gg 1$, the matrix $\widehat{\bbs{M}}(\lambda;x,t)$ will have polar singularities very close to the real axis, while still having neighboring distances (i.e. horizontal spacing)  of order $\bigo{1}$: 
\[
    \lambda_k := \frac{ - v_k + \ii \mu_k}{2\Delta} =: \frac{-\widehat v_k + \ii \widehat \mu_k }2 \quad \text{with} \quad \widehat v_k-\widehat v_{k-1}\geq 1 \ .
\]
More rigorously, it is easy to verify that $\widehat{\bbs M}(\lambda;x,t)$ satisfies the following RHP:
\begin{enumerate}
	\item $\widehat{\bbs{M}}(\lambda; x,t)$ is analytic for $\lambda \in \C \setminus \{ \lambda_k, \overline\lambda_k \}_{k=1}^N$;
	\item $\widehat{\bbs{M}}(\lambda;x,t) = \bbs I + \bigo{\lambda^{-1}}$ as $\lambda \to \infty$;
	\item $\widehat{\bbs{M}}(\lambda;x,t)$ has a simple pole at each $\lambda_k$ and $\overline\lambda_k$ satisfying the residue relation
	\begin{equation}
	\begin{aligned}
		&\Res_{\lambda=\lambda_k} \widehat{\bbs{M}}(\lambda; x,t) = 
		\widehat\gamma_k(x,t)  \lim_{\lambda \to \lambda_k} \widehat{\bbs{M}}(\lambda;x,t) \tril[0]{1}, \\
		&\Res_{\lambda=\overline\lambda_k} \widehat{\bbs{M}}(\lambda; x,t) = -\overline{\widehat\gamma_k(x,t)} 
		\lim_{\lambda \to \overline{\lambda}_k} \widehat{\bbs{M}}(\lambda;x,t) \triu[0]{1},
	\end{aligned}
	\end{equation}
    where \(\widehat \gamma_k(x,t) = \gamma_k(x,t)/\Delta\) with $\gamma_k(x,t) := c_k \ee^{2\ii\theta(z_k;x,t)}$. 
\end{enumerate}
The solution of fNLS is then given by 
\begin{equation}\label{eq:recover.2}
    \psi_N(x,t) = \lim_{\lambda \to \infty} 2\ii \Delta \lambda \, [\widehat{\bbs{M}}(\lambda;x,t)]_{1,2}.
\end{equation}

In order to consider all values of $N\in \mathbb N$ simultaneously, we embed \( \widehat v_1,\ldots,\widehat v_N \) into a sequence \( \{\widehat v_k\}_{k\in\N}\), \( |\widehat v_k-\widehat v_j|\geq1\) for any \( k\neq j\) and let \( c_k =0 \) for all \( k>N \) (so that \( \widehat\gamma_k(x,t)=0 \), \( k>N \)). Then, for each \( k\in \N \) we define 
\[
\Gamma_k:= \big\{\lambda \in \C : | \lambda + \tfrac12\widehat v_k | = {\varepsilon_0} \big\}\ , 
\]
a circle centered at $-\tfrac12\widehat v_k$ and orient it counterclockwise for some $\varepsilon_0>0$ sufficiently small, so that $\Gamma_k\cap \Gamma_j = \emptyset$, $\forall j,k \in \N$, $j\neq k$.
We are assuming the parameter \( \Delta \) is large enough so that the pair of poles $\lambda_k$ and $\overline \lambda_k$ lies in the interior of $\Gamma_k$ for all \( k=1,\ldots,N \). We set $\Gamma := \bigcup_{k=1}^\infty \Gamma_k$.  
 
 In the next step we trade polar singularities for jump relations on \( \Gamma \). Thus, we define
 \begin{equation}
\label{eq:N.def}
	\bbs{N}(\lambda;x,t) := 
	\begin{dcases}
	\widehat{\bbs{M}}(\lambda;x,t) 
		\tril{a_k(\lambda)}
		\triu{b_k(\lambda)},
	&\lambda \in \inside(\Gamma_k),\\
	\widehat{\bbs{M}}(\lambda;x,t), 
	& \text{ otherwise},
	\end{dcases}
\end{equation}
where 
 \begin{equation}\label{eq:a_k/b_k}
 \begin{gathered}
 a_k(\lambda) = \frac{\ii \widehat\gamma_k(x,t)}{\widehat{\mu}_k} \frac{\lambda-\overline\lambda_k}{\lambda-\lambda_k}  
 = \left( \frac{\ii \widehat\gamma_k(x,t)}{\widehat{\mu}_k} - \frac{\widehat\gamma_k(x,t)}{\lambda-\lambda_k} \right)\ , 
 \\ 
 b_k(\lambda) = \frac{ \widehat{\mu}_k^2 \,\overline{\widehat \gamma_k(x,t)} }{\widehat{\mu}_k^2 + | \widehat \gamma_k(x,t) |^2} \frac{1}{\lambda - \overline{\lambda}_k}
 = -\frac1{2\ii\Delta} \frac{\psi^{(k)}(x,t)}{\lambda- \overline\lambda_k}\ .
 \end{gathered}
 \end{equation}

\begin{lemma}
    The matrix $\bbs N(\lambda;x,t) $ is analytic in $\C \setminus \Gamma$.
\end{lemma}
\begin{proof}
Write \( \widehat{\bbs M} = [\widehat{\bbs M}_1,\widehat{\bbs M}_2] \). Since
\[
\widehat{\bbs  M}\,\bbs  E_{21} = [\widehat{\bbs  M}_2,0], \quad \widehat{\bbs  M}\,\bbs  E_{12} = [0,\widehat{\bbs  M}_1], \quad \text{and} \quad \widehat{\bbs M} \, \bbs  E_{22} = [0,\widehat{\bbs M}_2],
\]
where \( \bbs E_{ij} \) is the \(2\times 2\) elementary matrix whose entries are zero except for the \((i,j)\)-th entry, which is 1, it holds that
\[
\bbs  N = \widehat{\bbs{M}} \big( \bbs I + a_k\bbs E_{21} + b_k\bbs E_{12} + a_kb_k \bbs E_{22} \big) = \big[\widehat{\bbs  M}_1+a_k\widehat{\bbs  M}_2,\widehat{\bbs  M}_2+b_k\widehat{\bbs  M}_1+a_kb_k\widehat{\bbs  M}_2 \big].
\]
The residue conditions of \( \widehat{\bbs M} \) can be rewritten as
\[
\begin{cases}
\underset{\lambda =\lambda_k}{\mathrm{Res}}\,\widehat{\bbs  M} = \big[\underset{\lambda =\lambda_k}{\mathrm{Res}}\,\widehat{\bbs  M}_1,0\big] =  \big[\widehat\gamma_k(x,t) \widehat{\bbs  M}_2(\lambda_k),0\big], \medskip \\
\underset{\lambda =\overline\lambda_k}{\mathrm{Res}}\,\widehat{\bbs  M} = \big[0,\underset{\lambda =\overline\lambda_k}{\mathrm{Res}}\,\widehat{\bbs  M}_2] =  [0,- \overline{\widehat\gamma_k(x,t)} \widehat{\bbs M}_1(\overline\lambda_k)\big].
\end{cases}
\]
Let \( \bbs N(\lambda) = [\bbs N_1(\lambda),\bbs N_2(\lambda)] \). We have that
\[
\begin{cases}
\underset{\lambda = \lambda_k}{\mathrm{Res}}\,\bbs  N_1 = \underset{\lambda = \lambda_k}{\mathrm{Res}}\,\widehat{\bbs  M}_1 + \underset{\lambda = \lambda_k}{\mathrm{Res}}\,a_k~\widehat{\bbs  M}_2(\lambda_k)  =0, \medskip \\
\underset{\lambda = \lambda_k}{\mathrm{Res}}\,\bbs N_2 = b_k(\lambda_k)\underset{\lambda = \lambda_k}{\mathrm{Res}}\,\big(\widehat{\bbs M}_1 + a_k\widehat{\bbs M}_2\big)  =0.
\end{cases}
\]
On the other hand, because \( a_k(\overline\lambda_k) =0 \), \( \bbs  N_1(\lambda) \) does not have a pole at \( \overline\lambda_k \). Moreover,
\begin{align*}
\underset{\lambda = \overline\lambda_k}{\mathrm{Res}}\,\bbs N_2 & = \left(1+a_k^\prime(\overline\lambda_k) \underset{\lambda = \overline\lambda_k}{\mathrm{Res}}\,b_k\right) \underset{\lambda = \overline\lambda_k}{\mathrm{Res}}\,\widehat{\bbs  M}_2 + \underset{\lambda = \overline\lambda_k}{\mathrm{Res}}\,b_k~\widehat{\bbs  M}_1(\overline\lambda_k) \\
& = \left( -\overline{\widehat{\gamma}_k(x,t)} + \left(1+ \frac{|\widehat\gamma_k(x,t)|^2}{\widehat\mu_k^2} \right)\underset{\lambda = \overline\lambda_k}{\mathrm{Res}}\,b_k \right)\ \widehat{\bbs M}_1(\overline\lambda_k) = 0.
\end{align*}

Hence, \( \bbs N \) is analytic in each \( \inside(\Gamma_k) \), but has a discontinuity across each \( \Gamma_k \) by construction.
\end{proof}

Instead of calculating directly the jump across $\Gamma$ for the matrix $\bbs N$, we first introduce one last transformation. We define
\begin{equation}\label{rhp:O_def}
	\bbs{O}(\lambda;x,t) := \begin{dcases}
		\bbs{N}(\lambda;x,t) \tril{-\frac{\ii \widehat \gamma_k(x,t)}{\widehat\mu_k} }, & z \in \inside(\Gamma_k),  \\
		\bbs{N}(\lambda;x,t),  & \text{otherwise}.
	\end{dcases}
\end{equation}

\begin{lemma}
\label{prop:rhpO}
    The matrix $\bbs{O}$ solves the following RHP:
\begin{enumerate}
	\item $\bbs{O}(\lambda; x,t)$ is analytic for $\lambda \in \C \setminus \Gamma$; 
	\item $\bbs{O}(\lambda;x,t) = \bbs I + \bigo{\lambda^{-1}}$ as $\lambda \to \infty$;
	\item $\bbs{O}(\lambda;x,t)$ has continuous boundary values on each side of $\Gamma$ satisfying the jump relations ($+$ and $-$ boundary values are on the right and left side of the contour when traversing it according to its orientation)
	\begin{align}
		\bbs{O}_+(\lambda;x,t) &= 
            \bbs{O}_-(\lambda;x,t)  \bbs{V}_{{O}}(\lambda;x,t) \ , \qquad \lambda \in \Gamma, \\
\bbs{V}_{{O}}(\lambda;x,t) 
  &= \bbs{I} + 
     \begin{multlined}[t][.7\linewidth]
       \frac{1}{2\ii \Delta} \sum_{k=1}^N \left( 
       \frac{1}{\lambda-\lambda_k} 
       \begin{bmatrix}
          0 & 0 \\ 
          -\overline{\psi^{(k)}(x,t)} & 
          -m^{(k)}(x,t)
       \end{bmatrix}
       \right. \\ \left.
       + \frac{1}{\lambda-\overline\lambda_k} 
       \begin{bmatrix}
         m^{(k)}(x,t) & 
         -\psi^{(k)}(x,t) \\
         0 & 0
       \end{bmatrix}
       \right) \mathbbm 1_k(\lambda),
     \end{multlined}
     \label{eq:V0.expand}
\end{align}
where $\psi^{(k)}(x,t)$ is the one-soliton solution \eqref{one_soliton} with scattering data $(\lambda_k,c_k)$, $m^{(k)}(x,t)$ is given by \eqref{eq:moment}, and $\mathbbm 1_k$ is the indicator function of $\Gamma_k$.
\end{enumerate}
\end{lemma}

\begin{proof} 
By construction (see \eqref{rhp:O_def}), $\bbs O$ satisfies points 1. and 2. of the lemma, and the jump condition reads 
	\begin{equation}
		\begin{aligned}
		\bbs{O}_+(\lambda;x,t) &= 
            \bbs{O}_-(\lambda;x,t)  \bbs{V}_{{O}}(\lambda;x,t), \\
		\bbs{V}_{{O}}(\lambda;x,t) &= \left(1-\sum_{k=1}^N\mathbbm 1_k(\lambda)\right)\boldsymbol I +
		\sum_{k=1}^N 
            \tril{
             a_k(\lambda)
            }
		\triu{
             b_k(\lambda) 
            }
		\tril{-\frac{\ii \widehat \gamma_k(x,t)}{\widehat\mu_k}}
		\mathbbm 1_k(\lambda). 
		\end{aligned}
	\end{equation}
Consider now the one-soliton problem with data $(z_k, c_k)$. The solution can be found by solving a $2\times 2$ linear system \eqref{alpha_beta}--\eqref{gamma}  which yields
\begin{equation}\label{eq:one.sol.first.moments}
\begin{gathered}
    \psi^{(k)}(x,t) = -2\ii \overline{\beta^{(k)}(x,t)} 
    \quad \text{and} \quad m^{(k)}(x,t) = -2\ii \alpha^{(k)}(x,t)\ ,\\
      \text{with} \ \alpha^{(k)}(x,t) = \ii \Delta \widehat \mu_k\, \frac{| \widehat \gamma_k(x,t)|^2}{\widehat \mu_k^2 + | \widehat \gamma_k(x,t)|^2 }\ ,
    \quad
    \beta^{(k)}(x,t) = \Delta \widehat\gamma_k(x,t)\, \frac{\widehat \mu_k^2}{\widehat \mu_k^2 + | \widehat \gamma_k(x,t)|^2 } \ .
    \end{gathered}
\end{equation}

Using \eqref{eq:one.sol.first.moments}, we calculate
\begin{equation*}
    \begin{aligned}
 a_k(\lambda) b_k(\lambda) & = 
 \ii \widehat \mu_k \frac{|\widehat\gamma_k(x,t)|^2}{\widehat{\mu}_k^2+|\widehat\gamma_k(x,t)|^2} \frac{1}{\lambda-\lambda_k} 
 = -\frac{1}{2\ii\Delta} \frac{m_k(x,t)}{\lambda-\lambda_k}, \medskip \\
\frac{\ii \widehat\gamma_k(x,t)}{\widehat \mu_k} b_k(\lambda) & = 
\ii \widehat \mu_k \frac{|\widehat\gamma_k(x,t)|^2}{\widehat{\mu}_k^2+|\widehat\gamma_k(x,t)|^2} \frac{1}{\lambda-\overline \lambda_k} 
 = -\frac{1}{2\ii\Delta} \frac{m_k(x,t)}{\lambda-\overline \lambda_k}, \medskip \\
a_k(\lambda)-\frac{\ii \widehat\gamma_k(x,t)}{\widehat \mu_k} \big(1+(a_kb_k)(\lambda) \big) &  
= \frac{\ii \widehat\gamma_k(x,t)}{\widehat \mu_k} \left( \frac{\lambda-\overline\lambda_k}{\lambda-\lambda_k} -1 -(a_kb_k)(\lambda)\right) 
= -\frac1{2\ii\Delta} \frac{\overline{\psi^{(k)}(x,t)}}{\lambda-\lambda_k}.
\end{aligned}
\end{equation*}
These relations readily yield \eqref{eq:V0.expand}. 
\end{proof}

We are now in a position to use the Small Norm Argument \cite{smallnormRH} and to derive the asymptotic behavior of the $N$-soliton potential in the regime $\Delta \gg 1$ and $\boldsymbol \mu$ bounded in $\ell^2$- and $\ell^\infty$-norms.

 The solution of the RHP $\bbs{O}$ (if it exists) can expressed in the form
\begin{equation}
\label{eq:O.integral.form}
	\bbs{O}(\lambda;x,t) = \bbs{I} + \int_{\Gamma} \frac{ (\bbs{I} + \boldsymbol{\eta}(\xi;x,t))(\bbs{V}_{O}(\xi;x,t) - \bbs{I}) }{\xi - \lambda} \frac{\mathrm{d} \xi}{2\pi \ii},
\end{equation}
where $\boldsymbol{\eta}$ is the solution of the integral equation
\begin{equation}\label{eq:eta_integral}
	(\mathbbm{1} - \mathcal{C}_{O})\boldsymbol{\eta} = \mathcal{C}_{O}[\bbs{I}], \qquad \mathcal{C}_{O}[\bbs f] := \mathcal{C}_-[\bbs  f (\bbs{V}_{O} - \bbs{I})] \ ,
\end{equation}
and $\mathcal C_-: L^2(\Gamma) \to L^2(\Gamma)$ is the Cauchy projection operator on $\Gamma$, namely
\begin{equation}\label{eq:Cauchy}
\mathcal C_-[\bbs f](\lambda) := \lim_{\substack{z \to \lambda \\ z \in \mbox{ \small right side of }  \Gamma}} \left( \frac{1}{2\pi \ii}\int_{\Gamma} \frac{\bbs f(\xi)}{\xi-z} \di \xi\right)\ .
\end{equation}
It remains to show that the integral operator $\mathbbm{1} - \mathcal{C}_{O}$ in \eqref{eq:eta_integral} is invertible, thus yielding existence (and
uniqueness) of the solution $\bbs O$.

\begin{lemma} 
\label{prop:small_norm}
The Cauchy operator $\mathcal C_O$ defined in \eqref{eq:eta_integral} has bounded norm
\begin{equation}\label{eq:CO_norm}
	\| \mathcal{C}_{O} \| \leq C_* \|\boldsymbol \mu \|_\infty \Delta^{-1}
\end{equation}
for some constant $ C_*>0$. Then, it follows for $\Delta > C_*\|\boldsymbol \mu \|_\infty$ that $\boldsymbol{\eta}$ exists and it  can be expanded as a convergent Neumann series
 \begin{equation}
     \boldsymbol{\eta} = (\mathbbm{1} - \mathcal{C}_{O})^{-1} \mathcal{C}_{O}[\bbs{I}] = \sum_{j=1}^\infty \mathcal{C}_{O}^{j}[\bbs{I}].
 \end{equation}
\end{lemma}
\begin{proof}
Using the estimates
\begin{equation}\label{eq:estimates}
|\psi^{(k)}(x,t)| \leq \mu_k \quad \text{and} \quad m^{(k)}(x,t) \leq \| \psi^{(k)}(\, \cdot\,,t) \|_{L^2(\mathbb R)}^2 = 2\mu_k \ , \qquad \forall \, (x,t)\in \mathbb R^2\ ,
\end{equation}
in the expression \eqref{eq:V0.expand} for the jump matrix $\bbs V_O$, it follows that 
    \begin{equation}
        \left\| (\bbs{V}_{O}(\lambda) - \bbs{I}) \big\vert_{\lambda \in \Gamma_{k}} \right\| 
        \  
        \begin{cases} 
          \leq C_0 \frac{\mu_k }{\Delta} , & 1\leq k \leq N, \\
          = 0, & k > N,
        \end{cases}
    \end{equation}
    for some constant $C_0>0$, where \( \|\boldsymbol f\|^2 := \operatorname{Tr}(\boldsymbol f^*\boldsymbol f) \) for a given matrix \( \boldsymbol f\) (if fact, one can use any matrix norm). Therefore, there exists a constant $C_1>0$ such that
\begin{equation}
\label{eq:VO}
        \| \bbs{V}_{O} - \bbs{I} \|_{L^\infty(\Gamma)} \leq C_1 \| \boldsymbol\mu \|_\infty \Delta^{-1}
        \quad \text{and} \quad
        \| \bbs{V}_{O} - \bbs{I} \|_{L^2(\Gamma)} \leq C_1 \| \boldsymbol\mu \|_2\Delta^{-1}
\end{equation}  
uniformly for all $(x,t) \in \R^2$, where \( \|\boldsymbol f\|_{L^\infty(\Gamma)} :=  \text{ess}\sup_{\Gamma} \|\boldsymbol{f}(\xi)\|\)  and \( \|\boldsymbol f\|_{L^2(\Gamma)}^2 := \int_\Gamma \|\boldsymbol f(\xi)\|^2 |\di \xi|\) for a given \( 2\times 2\) matrix-functions \( \boldsymbol f(\xi),\boldsymbol g(\xi)\). Since the loops $\Gamma_k$ have fixed radii and are well separated, the contour $\Gamma$ is Ahlfors-David regular\footnote{We recall that a set $G$ is Ahlfors-David regular if there exists $c,C>0$ such that $cr \leq \mathcal H^1(G\cap B_r(z)) \leq Cr$ for any $z\in \C$, $r\in (0,\operatorname{diam}G)$, where $\mathcal H^1$ is the 1-dimensional Hausdorff measure and $B_r(z)$ is the open ball centered at $z$ with radius $r$; see \cite{MR3892469}.} and it follows \cite{MR3892469,MR1405945} that the operator norm $\| \mathcal{C}_- \|$ is finite. The norm estimates \eqref{eq:VO} then yield that
\begin{align*}
	\left\| \mathcal{C}_{O}[\bbs f] \right\|_{L^2(\Gamma)} & \leq \| \mathcal{C}_- \| \left\| (\bbs V_O - \bbs I) \right\|_{L^\infty(\Gamma)} \left\|\bbs f  \right\|_{L^2(\Gamma)} \leq C_* \|\boldsymbol \mu \|_\infty \Delta^{-1} \left\|\bbs f  \right\|_{L^2(\Gamma)}, \smallskip \\
    \| \mathcal{C}_{O}[\bbs I] \|_{L^2(\Gamma)} & \leq \| \mathcal{C}_- \| \|\bbs V_O - \bbs I\|_{L^2(\Gamma)} \leq C_* \|\boldsymbol\mu\|_2\Delta^{-1},
\end{align*}
where \( C_* := \| \mathcal{C}_- \|C_1 \). Hence, for any $\Delta>C_*\|\boldsymbol\mu\|_\infty$, $(\mathbbm{1} - \mathcal{C}_{O})^{-1}$ exists and $\boldsymbol{\eta}$ can be expanded as a convergent Neumann series
 \[
     \left\| \boldsymbol{\eta}\right\|_{L^2(\Gamma)} = \bigg\| \sum_{j=1}^\infty \mathcal{C}_{O}^{j}[\bbs{I}] \bigg\|_{L^2(\Gamma)} \leq \sum_{j=1}^\infty \|\mathcal C_O\|^{j-1}\big\|\mathcal C_O[\bbs I]\big\|_{L^2(\Gamma)} \leq \frac{C_* \|\boldsymbol\mu\|_2}{\Delta-C_*\|\boldsymbol\mu\|_\infty} . \qedhere
 \]
 \end{proof}

From \eqref{eq:V0.expand}, we can now derive an explicit expression for $\mathcal{C}_{O}[\bbs{I}]$:
\begin{equation}
    \begin{aligned}
    \mathcal{C}_{O}[\bbs{I}](\lambda) 
    &= \sum_{k=1}^N \frac{1}{2\pi \ii}\oint_{\Gamma_k} (\bbs{V}_{O} (\xi) - \bbs{I}) \frac{\di \xi}{\xi-\lambda} \\
    &= \begin{multlined}[t][.8\textwidth]
    \frac{1}{2\ii \Delta} 
    \sum_{k=1}^N \frac{1}{2\pi \ii}\oint_{\Gamma_k}
    \left( 
       \frac{1}{\xi-\lambda_k} 
       \begin{bmatrix}
          0 & 0\\ 
          -\overline{\psi^{(k)}(x,t)} & 
          -m^{(k)}(x,t)
       \end{bmatrix} 
       \right. \\ \left.
       + \frac{1}{\xi-\overline\lambda_k} 
       \begin{bmatrix}
         m^{(k)}(x,t) & 
         -\psi^{(k)}(x,t) \\
         0 & 0
       \end{bmatrix}
       \right) 
    \frac{\di \xi}{\xi-\lambda},
    \end{multlined}
    \end{aligned}
\end{equation}
where $\lambda$ is understood to lie outside each of the $\Gamma_k$. Evaluating by residues gives
\begin{equation}
\label{CO}
     \mathcal{C}_{O}[\bbs{I}](\lambda) 
    = \frac{1}{2\ii \Delta} \sum_{k=1}^N 
    \left( 
    \frac{1}{\lambda-\lambda_k} 
    \begin{bmatrix}
        0 & 0\\ 
        \overline{\psi^{(k)}(x,t)} & 
        m^{(k)}(x,t)
    \end{bmatrix} 
    + \frac{1}{\lambda-\overline\lambda_k} 
    \begin{bmatrix}
       -m^{(k)}(x,t) & 
       \psi^{(k)}(x,t) \\
       0 & 0
    \end{bmatrix}
    \right).   
\end{equation}

Expanding \eqref{eq:O.integral.form} for large $\lambda$ gives
\[
\begin{cases}
    \bbs{O}(\lambda;x,t) &\displaystyle  = \bbs{I} + \frac{\bbs{O}_1(x,t)}{\lambda} + \bigo{\lambda^{-2}}, \medskip \\
    \bbs{O}_1(x,t) & \displaystyle  = -\frac{1}{2\pi \ii } \int_{\Gamma} (\bbs{I} + \boldsymbol{\eta}(\xi;x,t) )(\bbs{V}_{O}(\xi;x,t) - \bbs{I}) \di \xi  = \sum_{j=0}^\infty \bbs{O}_1^{(j)}(x,t), \medskip \\
    \bbs{O}_1^{(j)}(x,t) & \displaystyle := -\frac{1}{2\pi \ii } \int_{\Gamma} \mathcal{C}_{O}^{j}[\bbs{I}](\xi) 
 (\bbs{V}_{O}(\xi;x,t) - \bbs{I}) \di \xi.
 \end{cases}
\]
Computing the first two terms by residues  with the help of \eqref{CO} gives
\begin{align}
\label{eq:O1.expand.0}
    \bbs{O}_1^{(0)}(x,t) &= -\frac{1}{2\pi \ii} \int_{\Gamma} (\bbs{V}_{O}(\xi) - \bbs{I})\, \di \xi 
    = \frac{1}{2\ii \Delta} \sum_{k=1}^N 
    \begin{bmatrix}
        -m^{(k)}(x,t) 
        & \psi^{(k)}(x,t) \\[.5em]
        \overline{\psi^{(k)}(x,t)} & 
       m^{(k)}(x,t)
    \end{bmatrix} \\
   \bbs{O}_1^{(1)}(x,t) &= -\frac{1}{2\pi \ii} \int_{\Gamma} \mathcal{C}_{O}[\bbs{I}](\xi)(\bbs{V}_{O}(\xi) - \bbs{I})\, \di \xi \nonumber \\
    & = \frac{1}{(2\ii \Delta)^2} \sum_{j=1}^N \frac{1}{2\pi \ii} \oint_{\Gamma_j} 
    \sum_{k=1}^N  
    \begin{bmatrix}
        \frac{-m^{(k)}}{\xi-\overline\lambda_k} & 
        \frac{\psi^{(k)}}{\xi-\overline\lambda_k} \\[.5em] 
        \frac{\overline{\psi^{(k)}}}{\xi-\lambda_k} & 
       \frac{m^{(k)}}{\xi-\lambda_k}
    \end{bmatrix}   
    \begin{bmatrix}
       \frac{-m^{(j)}}{\xi-\overline\lambda_j}  & 
       \frac{\psi^{(j)}}{\xi-\overline\lambda_j}  \\[.5em] 
        \frac{\overline{\psi^{(j)}}}{\xi-\lambda_j} & 
        \frac{m^{(j)}}{\xi-\lambda_j}
       \end{bmatrix}
     \,\di \xi \nonumber \\
     & =\frac{1}{(2\ii \Delta)^2} \sum_{j=1}^N \sum_{\substack{ k =1 \\ k \neq j}}^N
     \begin{bmatrix}         
        \frac{\psi^{(k)}\overline{\psi^{(j)}}}{\lambda_j - \overline\lambda_k} + \frac{m^{(k)} m^{(j)}}{\overline\lambda_j - \overline\lambda_k} 
        &
          \frac{\psi^{(k)}m^{(j)}}{\lambda_j - \overline\lambda_k} - \frac{m^{(k)} \psi^{(j)}}{\overline\lambda_j - \overline\lambda_k}
        \\[.5em]
        \frac{ m^{(k)}\overline{\psi^{(j)}}}{\lambda_j -\lambda_k} - \frac{\overline{\psi^{(k)}} m^{(j)}}{\overline\lambda_j-\lambda_k}
        &
        \frac{m^{(k)} m^{(j)}}{\lambda_j - \lambda_k} 
        + \frac{\overline{\psi^{(k)}}\psi^{(j)}}{\overline\lambda_j - \lambda_k} 
     \end{bmatrix}  \ . \label{eq:O1.expand.1}
\end{align}
Furthermore, using the supremum matrix norm, it holds for any $j \geq 1$ that
\begin{align}
     \left\| \bbs{O}_1^{(j)}(x,t) \right\| & \leq \frac{1}{2\pi} \bigg\| \int_{\Gamma} \left| \mathcal{C}_{O}^j[\bbs{I}](\xi)(\bbs{V}_{O}(\xi) - \bbs{I}) \right| \, |\di \xi| \bigg\|
       \leq  \frac{1}{2\pi} \big\| \mathcal{C}_{O}^j[\bbs{I}] \big\|_{L^2(\Gamma)} \,
     \left\| \bbs{V}_{O}- \bbs{I} \right\|_{L^2(\Gamma)} \nonumber \\
     & \leq \frac1{2\pi}  \big\| \mathcal C_O \big\|^{j-1} \,
     \| \bbs{V}_{O}- \bbs{I} \|_{L^2(\Gamma)}^2 \leq \left( C_* \frac{\|\boldsymbol\mu\|_\infty}\Delta \right)^{j-1} \left(C_1 \frac{\| \boldsymbol\mu\|_2}\Delta \right)^2, 
\label{eq:Oj.estimate}     
\end{align}
Finally, undoing all the transformations and using formula \eqref{eq:recover.2}, we have that the $N$-soliton solution of \eqref{eq:nls} parameterized by scattering data \( \{(z_k,c_k)\}_{k=1}^N \) is given by
 \begin{equation}
     \psi_N(x,t) = \lim_{\lambda\to\infty} 2\ii\Delta \lambda [\widehat{\bbs M}(\lambda;x,t)]_{1,2} = \lim_{\lambda\to\infty} 2\ii \Delta \lambda [\bbs O(\lambda;x,t)]_{1,2} = 2\ii \Delta 
   \left[ \sum_{j=0}^\infty \bbs{O}_1^{(j)}(x,t)    \right]_{1,2} \ ,
\end{equation}
which, together with \eqref{eq:Oj.estimate}, gives \eqref{linear interaction}.

\subsection{Proof of Corollary~\ref{cor:sinc}}

We get from \eqref{eq:soliton parameters} and \eqref{eq:scattering_data} that
\[
\ii \mu e^{\mu x_k+\ii\phi_k} = c_k = \ii \mu e^{-v_k\mu + \ii(\mu^2-v_k^2)/2}.
\]
That is, \( x_k=-v_k\) and \( \phi_k=(\mu^2-v_k^2)/2\). It follows that
\[
\psi^{(k)}(x,t) = -\mu\sech\big(\mu(x-v_k(t-1)\big)e^{\ii\big(xv_k-\frac{t-1}2(v_k^2-\mu^2)\big)}
\]
by \eqref{one_soliton}. 
Hence, using \eqref{linear interaction} and Remark~\ref{error terms} we get that
\begin{multline}\label{eq:psiN.rescaled}
\frac1N \psi_N\left( \frac{2X}{N V},1+\frac{T}{(N V)^2} \right) = \\
-\frac\mu N \sum_{k=1}^N \sech\left(\frac{2\mu}{N V} \left(X-\frac{v_k}{N V} \frac{T}{2}\right)\right) e^{\ii \frac{\mu^2T}{2(N V)^2}} e^{\ii\big(2X \frac{v_k}{N V} - \frac T2 \frac{v_k^2}{(N V)^2} \big)} + \bigo{\frac1{\Delta}}.
\end{multline}
Setting \( s_k := v_k/(N V) \in [\alpha+\tfrac{k-1}{N},\alpha+\tfrac kN]\) for each \( k=1,\ldots,N \), it holds locally uniformly with respect to \( X,T\) that
\[
\frac1N \psi_N\left( \frac{2X}{N V},1+\frac{T}{(N V)^2} \right) = -\mu\sum_{k=1}^N e^{\ii\big(2Xs_k-\frac T2s_k^2\big)}\frac1N + \bigo{\frac1{\Delta}},
\]
where, for the error bound, we note that $\Delta \leq  V$. Standard error estimates of Riemann sum approximation now give that
\[
\frac1N \psi_N\left( \frac{2X}{N\Delta},1+\frac{T}{(N\Delta)^2} \right) = -\mu \int_{\alpha}^{\alpha+1} e^{\ii\big(2Xs-\frac T2 s^2\big)}\di s + \bigo{\max\left\{\frac1N,\frac1{\Delta}\right\}}.
\]
It only remains to notice that
\[
\int_{\alpha}^{\alpha+1} e^{\ii\big(2Xs-\frac T2 s^2\big)}\di s = e^{\ii\big(2\alpha X- \alpha^2\frac T2\big)}\int_0^1 e^{\ii\big(2\big(X-\frac {\alpha T}2\big)s-\frac T2 s^2\big)}\di s.
\]

\section{Stochastic $N$-soliton Solutions} \label{sec:Nsolitons_random}

We now consider $N$-soliton solutions whose  scattering data are random satisfying Assumption \ref{hp:random_scattering_intro}. 

As we are assuming the imaginary part of the poles $\{\mu_k = 2\Im(z_k)\}$ to be distributed as a sub-exponential random variable (see \cite[Definition 2.7]{Wainwright2019}), we recall some of its properties.

\begin{definition}\label{def:subexp}
    Given $\alpha, \nu \in \R_+$, a random variable $\fX$ is sub exponential of parameters $(\nu,\alpha)$ if 
    \begin{equation}
\label{subexp distr}
    \meanval{\ee^{\lambda(\fX - \E[\fX])}} \leq \ee^{\frac{\nu^2\lambda^2}{2}}\,, \quad \forall \, \lambda < \frac{1}{\alpha}\ .
\end{equation}
\end{definition}
In particular, this implies exponential decay of the tail of the distribution (see \cite[Proposition 2.10]{Wainwright2019}):
\begin{equation}
\label{eq:subexp_prop}
    \mathbb{P}\lp| \fX - \E[\fX] | \geq s \rp \leq 
        2 \exp \lp -\frac{1}{2} \min\left\{  \frac{s^2}{\nu^2}, \frac{s}{\alpha} \right\} \rp\,, \qquad \forall\, s>0 \,.
\end{equation}

As an example, we can consider the amplitudes to be distributed as a chi-squared distribution $\mu_k \sim \chi^2(\beta)$, $\beta \in \mathbb R_+$, i.e.
\begin{equation}
    \meanval{f(\mu_k)} = \frac{1}{2^{\frac{\beta}{2} }\Gamma\lp \frac{\beta}{2}\rp}\int f(\mu_k)\mu_k^{\frac{\beta}{2}-1}\ee^{-\frac{\mu_k}{2}}d\mu_k\,,
\end{equation}
where $\Gamma(\beta)$ is the Gamma-function \cite[Ch. 5]{DLMF}.

The $N$-soliton setup we are considering is similar to the deterministic case in Proposition~\ref{thm:darboux_intro} and Theorem~\ref{thm:1}: the velocities of the solitons are tuned so that at finite time ($t=1$) the solution will display a peak  of order $\bigo{N}$, due to all the solitons colliding together. \\
We  recall that provided that $\Delta > C_* \| \bm \mu \|_\infty$, the $N$-soliton solution can be written as 
\begin{equation}\label{eq:sum_sols}
\psi_N(x,t) = \sum_{k=1}^N \psi^{(k)}(x,t) + \wt \psi(x,t) + \bigo{\frac{\|\boldsymbol \mu\|_{\infty} \| \boldsymbol \mu\|_2^2}{\Delta^2}}\,, 
\end{equation}
with 
\begin{align}
&\psi^{(k)}(x,t) = - \mu_k\sech \lp \mu_k (x -k\Delta (t-1) ) \rp \ee^{\ii \lp k\Delta x + \tfrac{1}{2}(\mu_k^2 - k^2 \Delta^2)(t-1) \rp} \ , \\
\label{eq:second_order}
&\wt \psi(x,t) := \frac{1}{2\Delta \ii} \sum_{j=1}^N \sum_{\substack{ k=1 \\ k \neq j}}^N
   \left[ \frac{-m^{(k)}(x,t) \psi^{(j)}(x,t)}{\overline{\lambda_j} - \overline{\lambda_k}} 
    + \frac{\psi^{(k)}(x,t){m^{(j)}(x,t)}}{\lambda_j - \overline{\lambda_k}}
    \right] \ ,
\end{align}
where we set $z_j =  \lambda_j \Delta$ (see Theorem~\ref{thm:1}).

\subsection{Probability estimates}

In order to apply Theorem~\ref{thm:1} in a stochastic setting and obtain a CLT-type result, we will need to control the subleading terms of \eqref{eq:sum_sols} and we will approximate the leading term in an appropriate way. \\ 
We start by deriving a probabilistic bound for $\widetilde \psi$ \eqref{eq:second_order}.
\begin{lemma}\label{prop:prob_error_estimate}
Under Assumption \ref{hp:random_scattering_intro}, there exists a constant $K_\cD$, dependent on the distribution $\cD$, but independent of $N$, such that for any $ 0 \leq \varepsilon < \gamma - \frac{1}{2}$, 
\begin{equation}
\label{eq:delta_uniform}
    \P\left( \big| \wt \psi(x,t) \big| > N^{\frac{1}{2} -\varepsilon}\right) \leq K_\cD \frac{\ln(N)}{N^{\gamma - \frac{1}{2} - \varepsilon}}\, , \qquad \meanval{ \big| \wt \psi(x,t) \big|} \leq K_\cD \frac{\ln(N)}{N^{\gamma-1}}\,.
\end{equation}
where $\wt \psi(x,t)$  has been defined in \eqref{eq:second_order}.
\end{lemma}
\begin{proof}
First, we notice that, by exchanging $j,k$, in \eqref{eq:second_order}, we can rewrite the general term of the double sum as
\begin{equation}
   m^{(k)}\psi^{(j)}\left(\frac{2 \wo{\lambda_j} - \lambda_k - \wo{\lambda_k} }{(\wo{\lambda_j} - \wo{\lambda_k})(\lambda_k - \wo{\lambda_j})} \right) =  m^{(k)} \psi^{(j)}\left( \frac{-2}{(\wo{\lambda_j} - \wo{\lambda_k})} + \frac{\lambda_k - \wo{\lambda_k} }{(\wo{\lambda_j} - \wo{\lambda_k})(\lambda_k - \wo{\lambda_j})} \right)\,.
\end{equation}
Using $| \wo{\lambda_j} - \wo{\lambda_k}| \geq |k-j|$ and estimates \eqref{eq:estimates}, the expected value of the first term is bounded by
\begin{equation}\label{eq:error_estimate_prob1}
    \meanval{\left \vert m^{(k)}(x,t)\psi^{(j)}(x,t)\frac{2}{(\wo{\lambda_j} - \wo{\lambda_k})} \right\vert } \leq \frac{2 \meanval{\left| m^{(k)}(x,t)\right| \, \left|\psi^{(j)}(x,t)\right|}}{|k-j|} \leq \frac{4\mu^2_{\cD}}{|k-j|}\ ,
\end{equation}
where we also used independency of the $\mu_k$'s. Analogously, the expected value of the second term can be estimated as
\begin{equation}\label{eq:error_estimate_prob2}
    \meanval{\left \vert m^{(k)}(x,t)\psi^{(j)}(x,t)\frac{\lambda_k - \wo{\lambda_k} }{(\wo{\lambda_j} - \wo{\lambda_k})(\lambda_k - \wo{\lambda_j})} \right\vert } \leq \frac{\meanval{\left| m^{(k)}(x,t)\right| \, \left| \psi^{(j)}(x,t)\right| \mu_k}}{\Delta |k-j|^2} \leq \frac{2 \mu_{\mathcal D} \meanval{ \mu^2_k}}{\Delta |k-j|^2}\ .
\end{equation}
Therefore, from \eqref{eq:error_estimate_prob1}-\eqref{eq:error_estimate_prob2}, there exists a constant $K_{\cD}$ independent of $N,\Delta$, but depending on the distribution of the $\mu_k$ such that
\begin{equation}
\begin{split}
    \meanval{ \big| \wt \psi(x,t) \big|} & \leq \frac{1}{\Delta}\sum_{j=1}^N \sum_{\substack{ k=1 \\ k \neq j}}^N  \left(\frac{2\mu^2_{\cD}}{|k-j|} + \frac{\mu_{\mathcal D} \meanval{ \mu^2_k}}{\Delta |k-j|^2}\right)\ \leq \widetilde K_{\cD}\frac1\Delta \sum_{j=1}^N \sum_{\substack{ k=1 \\ k \neq j}}^N \left(\frac{1}{|k-j|} + \frac{1}{\Delta |k-j|^2} \right)\\& \leq \frac{K_\cD}{2} \left(\frac{N\ln(N)}\Delta + \frac{N}{\Delta^2} \right) \leq K_\cD\frac{N\ln(N)}{\Delta}  \,.
    \end{split}
\end{equation}
The statement of the Lemma  now follows by applying the Markov inequality that states that for any positive random variable $X$ and positive number $\lambda$, $\mathbb{P}(X>\lambda)<\frac{\meanval{X}}{\lambda}$ and then substituting $\Delta = \beta N^\gamma$.
\end{proof}


Our goal is to study the fluctuations of the solutions $\psi_N(x,t)$ in a neighbourhood of the collision singularity (i.e. $x=0$ and $t=1$). Therefore, we rescale the space and time variables
\[
x = \frac{2X}{\Delta N} 
\qquad \text{ and } \qquad t-1 = \frac{T}{\Delta^2N^2} \ ;
\]
note that $T<0$ indicates a time before the collision and $T>0$ indicates a time after the collision. \\
In a slight abuse of notation, we will express functions of the original variables $x$ and $t$ as functions of $X$ and $T$.  For example, $\psi_{N}(X,T)$ will be used instead of $\psi_{N}\left( \frac{2X}{\Delta N}, 1 +  \frac{T}{\Delta^2N^2} \right)$, and so on.

In the next Proposition we will show that the rescaled profile near the collision singularity 
\begin{equation}
\psi_N(X,T) =  - \sum_{k=1}^N \mu_k\sech \left( \mu_k \lp \frac{2  X}{\Delta N} - \frac{kT}{\Delta N^2}\rp \right) \ee^{\ii  \lp 2  \frac{k}{N}X - \tfrac{k^2}{2N^2}T   + \frac{\mu^2_k T}{2\Delta^2 N^2} \rp} + \wt \psi(X,T) +\bigo{\frac{\|\boldsymbol \mu\|_{\infty} \| \boldsymbol \mu\|_2^2}{\Delta^2}}
\end{equation}
can be approximated with high probability by the following function
\begin{equation}\label{eq:approx_NSS}
\widehat \psi_{N}(X,T) :=   -\sum_{k=1}^N \mu_k \ee^{\ii  \lp 2  \frac{k}{N}X - \tfrac{k^2}{2N^2}T   \rp}\ .
\end{equation} 
provided that $\Delta$ is large enough (i.e. we are in the small norm setting of Theorem \ref{thm:1}). 

To make this statement quantitative, we define the function $f_N(X,T)$ as

\begin{equation}
\label{eq:f_N}
    f_N(X,T) =  - \sum_{k=1}^N \mu_k\sech \left( \mu_k \lp \frac{2  X}{\Delta N} - \frac{kT}{\Delta N^2}\rp \right) \ee^{\ii  \lp 2  \frac{k}{N}X - \tfrac{k^2}{2N^2}T   + \frac{\mu^2_k T}{2\Delta^2 N^2} \rp} + \wt \psi(X,T)\,,
\end{equation}
then the following holds:

\begin{proposition}\label{prop:psi_approx}

Under Assumption \ref{hp:random_scattering_intro}, fix $(X,T)$ in a compact set and let $N>2$, then there exist two constants $C_\cD$, $ \hat C_{\mathcal{D}}$, depending on the distribution $\cD$, such that
\begin{equation}
\label{eq:approx_local}
    \meanval{ \Big|  f_N(X,T) - \widehat \psi_{N}(X,T)  \Big|} \leq C_\cD \frac{\ln(N)}{N^{\gamma - 1}} +   \hat C_{\mathcal D}\frac{(X^2 +T^2+|X T| + |T|)}{\beta^2N^{2\gamma + 1}} \ ,
\end{equation}
where $f_N(X,T)$ is defined in \eqref{eq:f_N} and $\widehat \psi_{N}(X,T)$ is defined in \eqref{eq:approx_NSS}.
\end{proposition}
\begin{proof}
By linearity of the expected value
\begin{equation}
       \meanval{ \Big|  f_N(X,T) - \widehat \psi_{N}(X,T)  \Big|}\leq  \meanval{\left\vert \sum_{k=1}^N \psi^{(k)}(X,T) - \widehat \psi_{N}(X,T) \right\vert}  + \meanval{\vert\wt \psi(X,T)\vert} \,.
\end{equation}
 From Lemma \ref{prop:prob_error_estimate}, we bound the second term  as 
\begin{equation}
      \meanval{\vert\wt \psi(X,T)\vert}
   \leq  K_\cD \frac{\ln(N)}{ N^{\gamma -1}}
\end{equation}
for some constant $K_\cD>0$ depending on the distribution $\cD$. Next
\begin{equation}
\label{eq:markov}
    \begin{split}
          &  \meanval{\left|\sum_{k=1}^N\psi^{(k)}(X,T) - \widehat \psi_{N}(X,T) \right\vert } \\ 
          &   \leq  \left( \meanval{\left|\sum_{k=1}^N \psi^{(k)}(X,T) - \sum_{k=1}^N \wt{\psi^{(k)}}(X,T) \right\vert }  + \meanval{\left|\sum_{k=1}^N \wt{\psi^{(k)}}(X,T) - \widehat \psi_{N}(X,T) \right\vert }\right)\, ,
    \end{split}
\end{equation}
where
    \begin{equation}
        \wt{\psi^{(k)}}(X,T): =  - \mu_k\sech \left( \mu_k \lp \frac{2  X}{\Delta N} - \frac{kT}{\Delta N^2}\rp \right) \ee^{\ii  \lp 2  \frac{k}{N}X - \tfrac{k^2}{2N^2}T  \rp}\,.
    \end{equation}

The first term is estimated as follows
   \begin{eqnarray}
       \meanval{\left|\sum_{k=1}^N \psi^{(k)}(X,T) - \sum_{k=1}^N \wt{\psi^{(k)}}(X,T) \right\vert } &\leq& \sum_{k=1}^N \meanval{\left| \wt{\psi^{(k)}}(X,T) \right|\, \left| 1 - \ee^{\ii \frac{\mu^2_k T}{2\Delta^2 N^2} } \right|} \nonumber  \\ & 
        \leq& 2\sum_{k=1}^N \E\left[\mu_k^3  \right] \frac{|T|}{2\Delta^2  N^2}\leq \frac{c_{1,\cD} |T|}{\Delta^2 N} \label{eq:first_approx}
    \end{eqnarray}
for some constant $c_{1,\cD}$,  where in the second inequality we used $\left\vert 1- \ee^{s}\right\vert \leq 2|s|$ for $s \in i \mathbb{R}$.
The second term can be easily bounded in a similar way
    \begin{eqnarray}
          \meanval{\left|\sum_{k=1}^N \wt{\psi^{(k)}}(X,T) - \widehat \psi_{N}(X,T) \right| } &\leq& \sum_{k=1}^N \meanval{\left|\ii \mu_k \ee^{\ii  \lp 2  \frac{k}{N}X - \tfrac{k^2}{2N^2}T   \rp} \right| \, \left| 1 - \operatorname{sech} \left( \mu_k \lp \tfrac{2  X}{\Delta N} - \tfrac{kT}{\Delta N^2}\rp \right) \right| } \nonumber \\
          &\leq& \frac{c_{2,\cD}}{\Delta^2 N} \lp X^2 + T^2 + |XT| \rp \, ,
           \label{eq:second_term}
    \end{eqnarray}
    for some constant $c_{2,\cD}$, 
where we used the inequality $ \left \vert 1 - \operatorname{sech}(s)\right\vert \leq s^2$ for $s\in\mathbb{R}$. 

   Finally, substituting $\Delta = \beta N^\gamma$ concludes the proof.
\end{proof}

We can now prove a CLT-type result for the approximating function $\widehat \psi_{N}(X,T)$. We will then be able to extend these results to the original $N$-soliton solution $\psi_N(x,t)$, using the results from Propositions \ref{prop:psi_approx}.

\begin{lemma}
\label{lem:CLT_approx}
    Fix $X,T\in \mathbb{R}$. Under Assumption \ref{hp:random_scattering_intro}, the following convergence results hold
    \begin{align}
    \label{eq:clt_re_approx}
        & \ \frac{\Re\lp \widehat \psi_{N}(X,T) - N \mu_\cD\psi_0(X,T)\rp}{ \sqrt{N \Var_{\cD}}} \xrightarrow[N\to\infty]{\text{law}}   \mathcal{N}\big( 0,\sigma_+(X,T) \big) \ , \\
        \label{eq:clt_im_approx}
          &  \frac{\Im\lp \widehat \psi_{N}(X,T) - N \mu_\cD\psi_0(X,T)\rp}{ \sqrt{N \Var_{\cD}} } \xrightarrow[N\to\infty]{\text{law}} \mathcal{N}\big( 0,\sigma_-(X,T) \big) \ , \\
          \label{eq:clt_mod_approx}
        &\frac{\left \vert  \widehat \psi_{N}(X,T) - N \mu_\cD\psi_0(X,T)  \right \vert}{ \sqrt{N \Var_{\cD}}} \xrightarrow[N\to\infty]{\text{law}}  \mathcal H\big(\varphi(X,T)\big)\ ,
    \end{align}
     where $\Var_{\cD}$ is the variance of the distribution $\mathcal D$,     
    \begin{equation}
    \psi_0(X,T) = - \int_0^1 \ee^{\ii\left( 2Xs - \tfrac{T}{2}s^2 \right)} \di s\ , \qquad \sigma_\pm(X,T) = \frac{1}{2}\left(1 \pm \int_0^1 \cos \left( 4Xs - Ts^2 \right) \right)\di s \ ,
    \end{equation}
    $\mathcal H(\varphi)$ is a special Hoyt distribution with probability density function 
    \begin{equation}
        \rho(\xi; \varphi) = \frac{2\xi\, {\rm e}^{-\frac{\xi^2 }{\sin^2(2\varphi)}}}{ | \sin (2 \varphi )| }  I_0\left(\xi^2 \cot (2 \varphi ) \csc (2 \varphi )\right)\ , \qquad \xi\in \R^+\ ,
    \end{equation}
    where $I_0$ is the modified Bessel function of first kind of order $\nu=0$ \cite[Formula 10.32.1]{DLMF}: 
\begin{equation}
I_0(z) 
=  \frac{1}{\pi} \int_0^\pi \ee^{\pm z\cos(\theta)} \di \theta = \frac{1}{\pi} \int_0^\pi \cosh(z\cos(\theta)) \di \theta \ .
\end{equation}
    and $\varphi\in \lb 0, \frac{\pi}{2}\rb$ is such that 

    \[ \cos(2\varphi) =   \int_0^1 \cos \left( 4Xs - Ts^2 \right) \di s\,. \]
    
\end{lemma}

\begin{proof}
We will resort to the Nagaev-Guivarc'h method, a fundamental technique to prove probabilistic limit theorems for dynamical systems, which we  briefly recall in Appendix \ref{app:probability_results}. 

We start by considering the real and imaginary parts of $\widehat \psi_{N}$ \eqref{eq:approx_NSS} separately.

It is enough to notice that 
\begin{align}
&\E \left[\ee^{- \ii \xi \Re\lp\widehat \psi_N \rp} \right] = \E \lb \ee^{ \ii \xi \lb \sum_{k=1}^N \mu_k  \cos \lp 2X \frac{k}{N}  - \frac{T}{2} \frac{k^2}{N^2} \rp \rb} \rb = \prod_{k=1}^N  \E \lb \ee^{ \ii \xi  y  \cos \lp  2X \frac{k}{N}  - \frac{T}{2} \frac{k^2}{N^2}  \rp } \rb = \prod_{k=1}^N \lambda_1(\xi; \tfrac{k}{N}) \ ,\\
 &\E \left[\ee^{-\ii \xi \Im\lp\widehat \psi_N \rp} \right] = \E \lb \ee^{ \ii \xi \lb \sum_{k=1}^N \mu_k  \sin \lp  2X \frac{k}{N}  - \frac{T}{2} \frac{k^2}{N^2} \rp \rb} \rb = \prod_{k=1}^N  \E \lb \ee^{\ii \xi  y  \sin \lp  2X \frac{k}{N}  - \frac{T}{2} \frac{k^2}{N^2} \rp } \rb = \prod_{k=1}^N \lambda_2(\xi; \tfrac{k}{N})\ , 
\end{align}
where $y\sim \mathcal{D}$ with mean $ \mu_\cD $ and variance $\Var_\cD$, and 
\begin{align}
\lambda_1(\xi;s) &= \operatorname{exp} \left\{ \log \E \lb e^{ \ii \xi  y  \cos \lp ( 2X s  - \frac{T}{2} s^2) \rp } \rb \right\} \\
&\nonumber = \operatorname{exp} \left\{  \ii \cos(2Xs - \tfrac{T}{2}s^2)\mu_\cD \xi - \cos^2(2Xs - \tfrac{T}{2}s^2) \,\Var_{\cD} \, \frac{\xi^2}{2} + o(\xi^{2}) \right\}\ , \\
\lambda_2(\xi;s) &= \operatorname{exp} \left\{ \log \E \lb e^{ \ii \xi  y  \sin \lp ( 2X s  - \frac{T}{2} s^2) \rp } \rb \right\} \\
& = \operatorname{exp} \left\{  \ii \sin(2Xs - \tfrac{T}{2}s^2) \mu_\cD \xi - \sin^2(2Xs - \tfrac{T}{2}s^2) \, \Var_{\cD}\, \frac{\xi^2}{2} + o(\xi^{2}) \right\}\  . \nonumber
\end{align}
We define 
\[
\psi_0(X,T) = -\int_0^1 \ee^{\ii \left(2Xs-\frac{T}{2}s^2\right)} \di s \ ,
\]
 and, by applying \cite[Theorem 4.2]{Mazzuca2024} (see also Appendix \ref{app:probability_results}), we directly obtain \eqref{eq:clt_re_approx}-\eqref{eq:clt_im_approx}.
 
The complex random variable $\hat \psi_{N}(X,T)$ has expected value $N\mu_\cD\psi_0(X,T)$, which implies  
\begin{equation}
    \lim_{N\to\infty } \frac{\widehat \psi_{N}(X,T)}{\mu_{\mathcal D} N} = \psi_0(X,T)\,\quad \text{almost surely} \ .
\end{equation}
 which is a \textit{universal profile}. 
From the previous calculations, it follows that the real random variable 
\[ \mathcal X_N(X,T):=  \frac{\left \vert \widehat \psi_{N}(X,T) - N \mu_\cD\psi_0(X,T) \right \vert}{\sqrt{N m_{2,\cD}}}\, \]
converges to the following probability distribution
\begin{equation}
       \mathcal X_N(X,T)
      \xrightharpoonup{N\to\infty} \mathcal{Z}(X,T) \ , \qquad 
      \mathcal{Z}(X,T) := \sqrt{\cU(X,T)^2 + \cV(X,T) ^2 } 
\end{equation}
pointwise in $(X,T)\in \R^2$,
where $\cU(X,T) \sim  \mathcal{N}\lp 0,\sigma_+(X,T) \rp$ and $\cV(X,T) \sim \mathcal{N}\lp 0,\sigma_-(X,T) \rp $, with
 \begin{gather}
     \sigma_\pm(X,T) = \frac{1}{2}\left(1 \pm \int_0^1 \cos \left( 4Xs - Ts^2 \right) \di s\right)\,. 
     \end{gather}
The cumulative distribution function of the modulus of a complex Gaussian is equal to
\begin{align}
\P \lp \left|\mathcal Z\right|\leq \xi \rp & = \P \lp \cU^2 + \cV^2  \leq  \xi^2 \rp =\frac{1}{2\pi \sqrt{\sigma_+(X,T) \sigma_-(X,T)}} \int_{u^2+v^2 \leq\xi^2}  \ee^{-\frac{1}{2} \left( \frac{u^2}{\sigma_+(X,T)}+ \frac{v^2}{\sigma_-(X,T)} \right)}  \di u \di v \nonumber \\
& = \frac{1}{|\sin(2\varphi) |\pi} \int_{u^2+v^2 \leq   \xi^2}  \ee^{- \frac{1}{2}\left( \frac{u^2}{\cos^2(\varphi)}+ \frac{v^2}{\sin^2(\varphi)} \right)}  \di u \di v , \qquad \forall \ \xi\geq 0\ ,
\end{align}
where we defined 
\[ \cos(2\varphi) =  \int_0^1 \cos \left( 4Xs - Ts^2 \right) \di s\,, \]
which implies that 
\[
\sigma_+(X,T) = \cos^2(\varphi) \qquad \text{and} \qquad \sigma_- (X,T)  =  \sin^2(\varphi)  \, .
\]
 From this expression, we can compute the probability density function (Hoyt distribution):
\begin{equation}
\rho(\xi; \varphi) = \frac{2 \xi\, {\rm e}^{-\frac{ \xi^2 }{ \sin^2(2\varphi)}}}{ | \sin (2 \varphi )| }  I_0\left(\xi^2 \cot (2 \varphi ) \csc (2 \varphi )\right),
\end{equation}
where $I_0$ is the modified Bessel function of first kind of order $\nu=0$ \cite[Formula 10.32.1]{DLMF}: 
\begin{equation}
\label{I0}
I_0(z) 
=  \frac{1}{\pi} \int_0^\pi \ee^{\pm z\cos(\theta)} \di \theta = \frac{1}{\pi} \int_0^\pi \cosh(z\cos(\theta)) \di \theta \ .
\end{equation}
This proves \eqref{eq:clt_mod_approx}.
\end{proof}

\subsection{Proof of the Central Limit Theorem \ref{thm:local_CLT_intro} and Proposition \ref{prop:prop_CLT_intro}.}

We will now prove Theorem \ref{thm:local_CLT_intro}, and   Proposition \ref{prop:prop_CLT_intro}, obtaining a CLT-type result for the solution $\psi_N(x,t)$ near the collision point $(x_0,t_0) = (0,1)$. 

We prove only the first limit \eqref{eq:clt_re}, since the proof in the other cases is analogous. The idea is to show since  $\psi_N(X,T)$ behaves like $\widehat \psi_N(X,T)$ in the limit as $N\to \infty$, with high probability, it obeys a CLT-type behaviour as well.

Consider the quantity
\begin{equation}
     \frac{\Re\Big( \psi_{N}(X,T)- N \mu_\cD\psi_0(X,T)\Big)}{\sqrt{N \Var_{\cD}}}  = \frac{\Re\lp\widehat \psi_{N}(X,T) - N \mu_\cD\psi_0(X,T)\rp}{ \sqrt{N \Var_{\cD}}}  + \frac{\Re\lp\psi_N(X,T) - \widehat \psi_{N}(X,T)\rp}{ \sqrt{N \Var_{\cD}}}\,.
\end{equation}
Thanks to Lemma \ref{lem:CLT_approx}, the first term will converge to a Gaussian with mean zero and prescribed variance. It remains to prove that there exists a $\delta>0$ such that
\begin{equation}
\label{eq:tail}
    \mathbb{P}\lp \frac{\left \vert \psi_N(X,T) - \widehat \psi_{N}(X,T)\right\vert }{ \sqrt{N \Var_{\cD}}} > N^{-\delta} \rp \xrightarrow{N\to\infty} 0\ ,
\end{equation}
for any $X,T\in \R$ and for any $\Delta >0$.

Fix $\frac{1}{2}>\delta>0$ and define the sets $\Omega$ and $A_\delta$ as 
\begin{align}\label{eq:Omega}
    &\Omega := \left\{\bm \mu \in \R^N_+ \, \Big\vert \, \,\,\|\bm \mu\|_\infty > N^{\frac{1}{3}} \,\right\} \,,\\
    &  A_\delta := \left \{ \bm \mu\in \R_+^N \, \Big \vert \, \left \vert \psi_N(X,T) - \widehat \psi_{N}(X,T)\right\vert > \sqrt{\Var_{\cD}} N^{\frac{1}{2} - \delta} \right\}  \subseteq \R_+^N\,.
\end{align}
Given $\Delta = \beta N^{\gamma}$ with $\gamma >\tfrac{1}{2}$, for $N$ big enough, $\Omega$ contains the set of vectors $\bm\mu$ for which $\|\bm \mu\|_\infty > \beta N^{\frac{1}{2}}$ (i.e. $\tfrac{\|\bm \mu\|_\infty }{\Delta}>1$) by construction; and we can rewrite \eqref{eq:tail} as

\begin{equation}
\label{eq:newclaim}\mathbb{P}\left(A_\delta\right)\xrightarrow{N\to\infty}0\,.
\end{equation}
Furthermore
\begin{equation}\label{eq:3.55}
    \mathbb{P}\left( A_\delta \right)  = \mathbb{P}\left(  A_\delta \cap \Omega \right) + \mathbb{P}\left(  A_\delta \cap (\R_+^N\setminus\Omega) \right)\,.
\end{equation}
One immediately notices (see Lemma \ref{lem:scaling_max} in Appendix \ref{app:probability_results}) that there exists a positive constant $c$ such that

\begin{equation}
\label{eq:no small norm}
    \mathbb{P}(A_\delta \cap \Omega)\leq \mathbb{P}(\Omega)\leq e^{-cN^{\frac{1}{4}}}\,.
\end{equation}
Therefore, we just need to estimate $\mathbb{P}(A_\delta\cap(\R_+^N\setminus\Omega))$. 
In this set, we can apply Theorem \ref{thm:1} and Proposition \ref{prop:psi_approx} therefore we conclude that there exist a function $\wt f(X,T,\boldsymbol{\mu},\Delta,N)$ and a constant $C$, independent of $(\boldsymbol{\mu},\Delta,N)$, such that

\begin{equation}
\label{eq:ftilde_bound}
\begin{split}
&    \psi_N(X,T) = f_N(X,T) + \wt f (X,T,\boldsymbol{\mu} ,\Delta,N)
    \mbox{ for all } \mu \in \mathbb{R}_{N}^{+}\setminus \Omega, \mbox{ and } \\
    &
    \qquad \vert \wt f(X,T,\boldsymbol{\mu},\Delta,N)\vert \leq C\frac{\|\boldsymbol \mu\|_{\infty} \| \boldsymbol \mu\|_2^2}{\Delta^2}\, \mbox{ for all } \mu \in \mathbb{R}_{+}^{N}.
\end{split}
\end{equation}
Indeed, one may use
\begin{eqnarray}
    \wt f(X,T,\boldsymbol{\mu},\Delta,N) = \begin{cases}
        \psi_{N}(X,T) - f_{N}(X,T) & \mbox{ if } C^{*} \| \boldsymbol \mu \|_{\infty} < \Delta \\
        \frac{C \| \boldsymbol \mu \|_{\infty} \| \boldsymbol \mu \|_{2}^{2}}{\Delta^{2}} & \mbox{ otherwise.} \\
    \end{cases}\ 
\end{eqnarray}
Then one can estimate $\mathbb{P}\lp A_\delta\cap(\R_+^N\setminus\Omega)\rp$ as follows.

\begin{equation}
\begin{split}
     \mathbb{P}\lp A_\delta\cap(\R_+^N\setminus\Omega)\rp& =  \mathbb{P}\lp \left\{\bmu\in \R^N_+ \Bigg\vert\frac{\left \vert \wh\psi_N(X,T) - f_N(X,T) - \wt f (X,T,\boldsymbol{\mu},\Delta,N)\right\vert }{ \sqrt{N \Var_{\cD}}} > N^{-\delta} \right\}\cap (\R_+^N\setminus\Omega)\rp \,\\
     & \leq \mathbb{P}\lp \left\{\bmu\in \R^N_+ \Bigg\vert\frac{\left \vert \wh\psi_N(X,T) - f_N(X,T) - \wt f (X,T,\boldsymbol{\mu},\Delta,N)\right\vert }{ \sqrt{N \Var_{\cD}}} > N^{-\delta} \right\}\rp\,.
     \end{split}
\end{equation}
By Markov inequality, that states that for a positive random variable $\fX$ and $\lambda>0$, $\mathbb{P}(\fX>\lambda)\leq \frac{\meanval{\fX}}{\lambda}$, we bound the last term as

\begin{equation}
    \begin{split}
        \mathbb{P}&\lp \frac{\left \vert \wh\psi_N(X,T) - f_N(X,T) - \wt f (X,T,\boldsymbol{\mu},\Delta,N)\right\vert }{ \sqrt{N \Var_{\cD}}} > N^{-\delta} \rp \leq \frac{\meanval{\left \vert \wh\psi_N(X,T) - f_N(X,T) - \wt f (X,T,\boldsymbol{\mu},\Delta,N)\right\vert }}{N^{\frac{1}{2} - \delta}\sqrt{\Var_\cD}}\\
        & \leq \frac{\meanval{\left \vert \wh\psi_N(X,T) - f_N(X,T) \right\vert} +\meanval{\left\vert \wt f (X,T,\boldsymbol{\mu},\Delta,N)\right\vert }}{N^{\frac{1}{2} - \delta}\sqrt{\Var_\cD}}\,.
    \end{split}
\end{equation}
We recall that for subexponential random variables, there exists a constant $\texttt{c}$ independent of $N$ such that $\meanval{\max_j\mu_j^3}\leq \texttt{c}\ln^3(N)$; therefore, by \eqref{eq:ftilde_bound} one deduces
\begin{equation}
    \meanval{\left\vert \wt f (X,T,\boldsymbol{\mu},\Delta,N)\right\vert} \leq C\frac{\ln^3(N)}{N^{2\gamma -1}}\,.
\end{equation}
So, applying Proposition \ref{prop:psi_approx}
we conclude that \eqref{eq:tail} holds and Theorem \ref{thm:local_CLT_intro} follows.
To prove Proposition \ref{prop:prop_CLT_intro}, one notices that once we have established \eqref{eq:clt_re}, i.e.
\begin{equation}
     \frac{\Re\Big( \psi_{N}(X,T) - N\mu_{\mathcal D}\, \psi_0(X,T)\Big)}{ \sqrt{N \Var_{\cD}}} \xrightarrow[N\to\infty]{\text{law}}  \mathcal{N}\big( 0,\sigma_+(X,T)\big) 
\end{equation}
it follows that 
\begin{equation}
  \frac{1}{\sqrt N}  \frac{\Re\Big( \psi_{N}(X,T) - N\mu_{\mathcal D}\, \psi_0(X,T)\Big)}{ \sqrt{N}} \xrightarrow{N\to\infty} 0  \qquad \text{ in probability,}
\end{equation}
since $\mathcal N(0,\sigma_+)$ is a continuous random variable. Finally, convergence in distribution to a constant implies convergence in probability, i.e.
\begin{equation}
    \lim_{N\to\infty} \mathbb P\left( \left| \Re\left( \frac{1}{N}\psi_{N}(X,T) - \mu_{\mathcal D}\, \psi_0(X,T)\right)\right| \geq \epsilon \right) =0 \quad \forall \ \epsilon>0\ .
\end{equation}

The same argument holds for the imaginary part of $\psi_N(X,T)$ (see $\eqref{eq:clt_im}$), thus implying that 
\begin{equation}
    \frac{1}{\mu_\cD N}\psi_N(X,T) \xrightarrow{N\to\infty} \psi_0(X,T) \qquad \text{in probability.} 
\end{equation}

\subsection{Proof of the Central Limit Theorem \ref{thm:envelope_intro}.}
Finally, we consider the general solution $\psi_N(x,1)$ at the collision time $t=1$ and we prove Theorem \ref{thm:local_CLT_intro}.
We first obtain a CLT-type result for the leading order term in the expansion \eqref{eq:sum_sols} at collision time $t=1$
\begin{equation}
\label{eq:approx_envelope}
   \varphi_N(x,1):= \sum_{k=1}^N \psi^{(k)}(x,1) = \sum_{k=1}^N \frac{- \mu_k}{\cosh \lp \mu_k x \rp} \ee^{\ii  k\Delta x}\,.
\end{equation}

\begin{lemma}
\label{lem:approx_envelope}
    Under Assumption \ref{hp:random_scattering_intro}, for any $x\in  \R$ the following convergence results hold
    \begin{align}
        & \frac{\Re\lp\varphi_N(x,1) \rp- \omega_\cD(x) \cos\lp x\Delta \frac{N+1}{2}\rp D_N\lp x\Delta\rp}{\sigma_{N,\Re}(x)} \xrightarrow[N\to\infty]{\text{law}}  \cN(0,1)\,, \label{eq:clt_envelope_approx_re} \\
        & \frac{\Im\lp \varphi_N(x,1) \rp - \omega_\cD(x) \sin\lp x\Delta \frac{N}{2}\rp D_{N+1}\lp x\Delta\rp}{\sigma_{N,{\Im}}(x)} \xrightarrow[N\to\infty]{\text{law}}  \cN(0,1)\,, \quad {\text{for } x\neq 0\,,} 
        \label{eq:clt_envelope_approx_im}
    \end{align}
     where $D_N(x) := \frac{\sin\lp \frac{x N}{2}\rp}{\sin \lp \frac{x}{2}\rp}$ is the Dirichlet kernel ,
    \begin{align}
        & \omega_\cD(x) = - \meanval{\frac{\xi}{\cosh( x\xi)}}\, , \quad \xi\sim \cD\, ,\\
        \label{eq:real_var_tilde}
        &\sigma^2_{N,\Re}(x) = \Var\lp\frac{\xi}{\cosh( x\xi)}  \rp\lp \frac{N-1}{2} + \frac{1}{2} \cos\lp x\Delta N \rp D_{N+1}(2x\Delta) \rp \,,  \\ 
        &\sigma^2_{N,\Im}(x) = \Var\lp\frac{\xi}{\cosh( x\xi)}  \rp\lp \frac{N+1}{2} - \frac{1}{2}  \cos\lp x\Delta N \rp D_{N+1}(2x\Delta) \rp \, ,
    \end{align}
    and $\Var(\cdot) $ is the variance of the given random variable. Moreover, 
    \begin{equation}\label{eq:clt_envelope_approx_mod}
        \frac{\left|\varphi_N(x,1)\right| - |\omega_\cD(x)D_{N}(x\Delta)|}{N} \to 0 \quad \text{as $N\to \infty$, in probability.}
    \end{equation}
\end{lemma}
\begin{remark}
    Notice that $\Im\lp\varphi_N(x,1) \rp =0$ deterministically, therefore we excluded the value $x=0$ in \eqref{eq:clt_envelope_approx_re}.
\end{remark}

\begin{proof} Let $x\in \R$. We will show the  proof of \eqref{eq:clt_envelope_approx_re} in detail. The proof of  \eqref{eq:clt_envelope_approx_im} is analogous.
The result easily follows from the classical result of Lyapounov's condition \cite{Billingsley2012}, which we recall in Appendix \ref{app:probability_results}.

Consider the real part of \eqref{eq:approx_envelope}:
$$\Re\lp \varphi_N(x,1) \rp = -\sum_{k=1}^N \mathcal X_k \cos(kx\Delta ) \, ,$$
where $\mathcal X_k$ is the random variable $\mathcal X_k(x) := \frac{\mu_k}{\cosh(\mu_k x)}$, $\mu_k\sim \cD$. 

Let $\sigma_{N,\Re}(x):= \sqrt{\sum_{k=1}^N \Var(\mathcal X_k)}$. We compute

    \begin{align}
        \sigma_{N,\Re}^4(x) & = \lp \sum_{k=1}^N \Var\lp \mathcal X_k \cos(kx\Delta ) \rp \rp^2= \Var\lp \mathcal X_1 \rp^2 \lp\sum_{k=1}^N  \cos^2(kx\Delta )\rp^2\, \nonumber \\
        & = \Var\lp \mathcal X_1\rp^2 \lp \frac{N}{2} - 1 + \frac{1}{2}\Re\lp \frac{1-\ee^{2\ii x\Delta(N+1)}}{1-\ee^{2\ii x\Delta}} \rp \rp^2 \, , \label{eq:denominator}
    \end{align}
since the amplitudes $\mu_k$'s are i.i.d., and
    \begin{align}
        \sum_{k=1}^N& \meanval{\Big( \mathcal X_k \cos(kx\Delta ) - \meanval{\mathcal X_k\cos(kx\Delta )} \Big)^4  }  
        = \meanval{\Big( \mathcal X_1 - \meanval{\mathcal X_1} \Big)^4}\sum_{k=1}^N \cos^4(kx\Delta ) \nonumber \\
         &= \meanval{\Big( \mathcal X_1 - \meanval{\mathcal X_1} \Big)^4} \lp \frac{3}{8}N - \frac{5}{8}+ \frac{1}{2} \Re\lp \frac{1-\ee^{2\ii x\Delta(N+1)}}{1-\ee^{2\ii x\Delta}} \rp +  \frac{1}{8} \Re\lp \frac{1-\ee^{4\ii x\Delta(N+1)}}{1-\ee^{4\ii x\Delta}} \rp\rp \,. \label{eq:numerator}
    \end{align}
Thus, the following Lyapounov's condition (with $\delta =2$) is satisfied
\[
\lim_{N\to \infty} \frac{1}{\sigma_{N,\Re}^4(x)} \sum_{k=1}^N \meanval{\Big( \mathcal X_k \cos(kx\Delta ) - \meanval{\mathcal X_k\cos(kx\Delta )} \Big)^4  }=0\ .
\]

Therefore, by the Lyapounov's condition Theorem \ref{thm:lyapounov}, the convergence result \eqref{eq:clt_envelope_approx_im}  follows

\begin{equation}
\frac{\Re\left(\varphi_N(x,1) - \meanval{ \varphi_N(x,1)}\right) }{\sigma_{N,\Re}(x)}   \xrightharpoonup{N\to\infty}  \cN(0,1)\ ,
\end{equation}
where
    \begin{equation}
       \meanval{\Re\left( \varphi_N(x,1)\right)}= -\meanval{\sum_{k=1}^N \mathcal X_k \cos(kx\Delta )} = 
       \omega_{\mathcal D}(x)\Re\lp \ee^{\ii x\Delta \frac{N+1}{2}} D_N(x\Delta)\rp\,.
    \end{equation}

Furthermore, since and $\sigma_{N,\Re}(x)$ is of the order $\bigo{N^{1/2}}$ for large $N$ ($\forall \, x\in\R$), Lyapounov's condition Theorem implies the following Law of Large Numbers result (see Proposition \ref{rem:almost_sure})
\begin{equation}
    \frac{1}{N} \left( \Re\lp\varphi_N(x,1)\rp - \omega_\cD(x)\cos\lp x\Delta \tfrac{N+1}{2}\rp D_N(x\Delta) \right)\to 0 \qquad \text{as $N \to \infty$, in probability;}
\end{equation}
similarly for the imaginary part of $\varphi_N(x,1)$. Therefore,  we can conclude that 
\begin{equation}
    \frac{\Big|\varphi_N(x,1)\Big| - |\omega_\cD(x)D_N(x\Delta)|}{N} \to 0 \qquad \text{as $N \to \infty$, in probability,}
\end{equation}
where we used the fact that 
 $ N^{-1}\lp \Big|\meanval{\varphi_N(x,1)}\Big| - \meanval{\Big|\varphi_N(x,1)\Big|}\rp  \to 0$, as $N\to\infty$,  and
    \begin{equation}
         \Big|\meanval{\varphi_N(x,1)}\Big| = \sqrt{\meanval{\Re(\varphi_N(x,1))}^2 + \meanval{\Im(\varphi_N(x,1))}^2 } \\ 
         = |\omega_\cD(x)D_N(x\Delta)| \,.
    \end{equation}
\end{proof}

Finally, in order to prove Theorem \ref{thm:envelope_intro}, we extend the result of the previous lemma to the solution $\psi_N(x)$ in the same way as in the proof of Theorem \ref{thm:local_CLT_intro}. We leave the details to the interested reader.

\paragraph{\bf Acknowledgements.} This project was made possible by a SQuaRE at the American Institute of Mathematics. The authors thank AIM for providing a supportive and mathematically rich environment, and excellent working conditions during the visit in Spring 2024 and 2025. 
We thank Gustavo Didier for the fruitful discussion about the CLTs and the references that he provided.

M.G. was supported in part by the National Science Foundation (grant no. DMS-2508767).
T.G. acknowledges the support of PRIN 2022 (2022TEB52W) "The charm of integrability: from
nonlinear waves to random matrices"-– Next Generation EU grant – PNRR Investimento M.4C.2.1.1 - CUP: G53D23001880006; the GNFM-INDAM group and the research project Mathematical Methods
in NonLinear Physics (MMNLP), Gruppo 4-Fisica Teorica of INFN. 
R.J. was supported in part by a grant from the Simons Foundation, CGM-853620 and in part by the National Science Foundation under Grant No. DMS-2307142. 
G.M. was partially supported by the Swedish Research Council under grant no. 2016-06596 while the author was in residence at Institut Mittag-Leffler in Djursholm, Sweden during the Fall semester 2024. 
M.Y. was supported in part by a grant from the Simons Foundation, CGM-706591. 

Most of the figures in the paper were realized using the Python libraries \texttt{NumPy} \cite{numpy}, \texttt{Scipy} \cite{Scipy} and \texttt{Matplolib} \cite{matplotlib}.

\appendix
\section[Facts about the $N$-soliton solution]{General facts about the $N$-soliton solution of the fNLS equation}\label{app:facts}
We report here some known facts about the RHP for $N$-soliton solutions, and we present a novel upper bound on the modulus of the solution $\psi_N(x,t)$.  Such a bound is suboptimal as compared to \eqref{eq:upper_bound_result} (see Theorem~\ref{thm:1}), but it shows  exponential decay of the tails. 
Given a set of spectral data $\{(z_k,c_k)\}_{k=1}^N$, the solution  \( \bbs M(z;x,t)\) of  the RHP \ref{rhp-sch} has the form (see \cite[Appendix B]{Jenkins2018})
\begin{equation}
\label{M_reflectionless}
    \bbs M(z;x,t) := \bbs I + \sum_{k=1}^N \frac1{z-z_k}\begin{bmatrix} \alpha_k(x,t) & 0 \smallskip \\ \beta_k(x,t) & 0  \end{bmatrix} + \sum_{k=1}^N \frac1{z-\overline z_k} \begin{bmatrix} 0 & -\overline{\beta_k(x,t)} \smallskip \\ 0 & \overline{\alpha_k(x,t)} \end{bmatrix},
\end{equation}
where \( \{\alpha_k(x,t),\beta_k(x,t)\}_{k=1}^N\) solve the following linear system
\begin{gather}
\label{alpha_beta} 
\alpha_k(x,t) = -\gamma_k(x,t) \sum_{j=1}^N\frac{\overline{\beta_j(x,t)}}{z_k-\overline z_j}
\quad \text{and} \quad
\beta_k(x,t) = \gamma_k(x,t) \left( 1 + \sum_{j=1}^N \frac{\overline{\alpha_j(x,t)}}{z_k-\overline z_j} \right), \\
\gamma_k(x,t) := c_k \ee^{2\ii\theta(z_k;x,t)}, \label{gamma}
\end{gather}
which follows directly from the residue conditions satisfied by \(\bbs M(z;x,t) \). 

\begin{proposition}
    The system \eqref{alpha_beta} is uniquely solvable, and the fNLS solution is given as 
    \begin{equation}
\label{reflectionless_sch}
\overline{\psi_N(x,t)} = 2\ii \sum_{k=1}^N \beta_k(x,t).
\end{equation}
\end{proposition}

\begin{proof} 
Unique solvability of system \eqref{alpha_beta} is equivalent to verify that
\begin{equation}
\label{matA}
\det\left(\bbs I + \bbs \Phi_N \overline{\bbs \Phi_N}\right) \neq 0, \quad \text{where } \   \bbs \Phi_N := \left[\frac{\sqrt{c_k} \sqrt{\overline {c_n} }}{z_k-\overline z_n}\ee^{\ii(\theta(z_k;x,t)-\theta(\overline z_n;x,t))}\right]_{k,n=1}^N,
\end{equation}
 \( \overline{\bbs \Phi_N} \) is the conjugate matrix of \( \bbs \Phi_N\), and  we consider the principal value of \( \sqrt {c_k}\). 
 
 We consider now the matrix $\ii \bbs \Phi_N$: since each entry  
 can be viewed as an inner product of linearly independent functions,
 \[
[\ii\bbs \Phi_N]_{k,n} = \int_0^\infty \big( \sqrt{c_k}\ee^{\ii\theta(z_k)} \big) \overline{\big( \sqrt{c_n}\ee^{\ii\theta(z_n)} \big)} \ee^{\ii(z_k-\overline z_n)s}\di s,
\]
\( \ii\bbs \Phi_N\) is a positive definite matrix. Let \( (\ii\bbs \Phi_N)^{1/2} \) be the unique positive definite square root of \( \ii\bbs \Phi_N\). The eigenvalues of \(\bbs \Phi_N \overline{\bbs \Phi_N} = (\ii\bbs \Phi_N)\overline{(\ii\bbs \Phi_N)}\) are the same as the eigenvalues of \((\ii\bbs \Phi_N)^{1/2} \overline{(\ii\bbs \Phi_N)} (\ii\bbs \Phi_N)^{1/2}\), which is also a positive definite matrix. If one labels these eigenvalues by \( \lambda_k>0 \), then
\[
\det\left(\bbs I + \bbs \Phi_N \overline{\bbs \Phi_N} \right)=\prod_{k=1}^N(1+\lambda_k)>0
\]
as needed. Finally, from \eqref{eq:recovery} and \eqref{M_reflectionless}, it  immediately follows that the corresponding solution \( \psi_N(x,t)\) of the fNLS equation \eqref{eq:nls} can be expressed as
\[
\overline{\psi_N(x,t)} = 2\ii \sum_{k=1}^N \beta_k(x,t). \qedhere
\]
\end{proof}

We derive now an alternative expression for the solution \( \psi_N(x,t)\), that will be used shortly to derive some estimates on the modulus of the solution.

\begin{lemma}
\label{prop:1}
Let \( B(z) \) be the Blaschke product from \eqref{eq:max_norming_const}. It holds that
\[
\psi_N(x,t) = 2\ii \sum_{k=1}^N \frac1{B^\prime(z_k)}\frac{\alpha_k(x,t)}{\gamma_k(x,t)},
\]
where $\alpha_k$ and $\gamma_k$ are defined in \eqref{alpha_beta} and \eqref{gamma} respectively.
\end{lemma}
\begin{proof}

It readily follows from \eqref{alpha_beta} that
\[
2\ii \sum_{n=1}^N \frac1{B^\prime(z_n)}\frac{\alpha_n(x,t)}{\gamma_n(x,t)} = \overline{2\ii \sum_{k=1}^N \left(\sum_{n=1}^N \frac1{\overline {B^\prime(z_n)}} \frac1{\overline z_n- z_k}\right) \beta_k(x,t)}.
\]
Since \( B(z) \) is a rational function with poles \( \overline z_1,\ldots,\overline z_N\) that is equal to \( 1 \) at infinity, it holds that
\[
B(z) = 1 + \sum_{n=1}^N \frac1{\overline{B^\prime(z_n)}} \frac1{z-\overline z_n}.
\]
As \( B(z_k)=0 \), it follows that \(\sum_{n=1}^N \frac1{\overline{B^\prime(z_n)}} \frac1{z_k-\overline z_n} = -1\), and we recover \eqref{reflectionless_sch}.
\end{proof}

The modulus of \( \psi_N(x,t) \) also admits expressions convenient for analysis. Indeed, it is known (see for example \cite[Equation (2.3)]{Jenkins2018}) that
\[
\partial_x\bbs M(z) = -\ii z\big[\sigma_3,\bbs M(z)\big] + \begin{bmatrix} 0 & \psi_N(x,t) \\ -\overline{\psi_N(x,t)} & 0 \end{bmatrix} \bbs M(z).
\]
Multiplying by \( z \) and taking the limit as \( z\to\infty \) of the \( (2,2) \)-entries of the above relation gives after conjugation that
\begin{equation}
\label{partial_alpha}
|\psi_N(x,t)|^2 = 2\ii \partial_x \left( \sum_{k=1}^N \alpha_k(x,t) \right).
\end{equation}
We recall that expression \eqref{partial_alpha} can further be rewritten using the famous determinantal formula \cite{faddeev2007hamiltonian}:
\[
|\psi_N(x,t)|^2 = \partial_{xx}\log\det\big(\bbs I + \bbs A \overline{\bbs A}\big) \ .
\]
We now present a novel, general upper bounds for \( |\psi_N(x,t)| \). Despite being suboptimal as compared to \eqref{eq:upper_bound_result} for finite $x,t \in \R\times\R^+$, it shows exponential decay for $|x|, |t| \gg 1$, which cannot be read from \eqref{eq:upper_bound_result}.

\begin{proposition}
\label{prop:3}
It holds that
\[
\frac12|\psi_N(x,t)| \leq \min \left\{\sum_{k=1}^N|\gamma_k(x,t)|, \;\;  \sum_{k=1}^N |B^\prime(z_k)|^{-1}, \;\; \sum_{k=1}^N |B^\prime(z_k)|^{-2}|\gamma_k(x,t)|^{-1} \right\}.
\]
\end{proposition}
\begin{proof}
Recall that if \( f(z) \) is analytic in a domain \( D \) then \( |f(z)|^2 \) is subharmonic there because \( \Delta |f(z)|^2 = 4\partial_z\partial_{\overline z}|f(z)|^2=4|f^\prime(z)|^2\geq 0 \). Let
\[
S_i(z) := \big| [\bbs M(z)]_{1,i} \big|^2 + \big| [\bbs M(z)]_{2,i} \big|^2, \quad i\in\{1,2\}.
\]
Since the second column of \( \bbs M(z) \) is analytic in \( \C_+ \), see \eqref{M_reflectionless}, \( S_2(z) \) is a subharmonic there. Clearly, \( \det \bbs M(z) \) is a rational function of \( z \) that is equal to \(1\) at infinity. Since only one column of \(\bbs M(z)\) can have a pole at a given point, \( \det \bbs M(z) \) can have at most simple poles. However, it is easy to check that the residue conditions for \( \bbs M(z) \) imply that the residues of \( \det\bbs M(z) \) are zero. Thus,  \( \det\bbs M(z)\equiv1 \). Since \( S_2(\infty) =1 \) and
\[
S_2(z) = [\bbs M(z)]_{2,2}(z)[\bbs M(z)]_{1,1} - [\bbs M(z)]_{1,2} [\bbs M(z)]_{2,1} = \det \bbs M(z) \equiv 1
\]
for \( z \) on the real line by \eqref{M_reflectionless}, the maximum principle for subharmonic functions implies that \( S_2(z) \leq 1 \) in \( \overline\C_+ \). 

These considerations can also be applied to the matrix \( \bbs M(z)B(z)^{\bbs \sigma_3} \), as it is still meromorphic with unit determinant. However, the roles of the columns are now reversed: the first one is analytic in \( \C_+ \) while the second one has poles therein. Thus, it now must hold that  \( (S_1|B|^2)(z) \leq 1 \) in \( \overline{\C_+} \). It readily follows from \eqref{M_reflectionless} and the residue conditions satisfied by \( \bbs M \) that
\[
\begin{bmatrix} \alpha_n \\ \beta_n \end{bmatrix} = \lim_{z\to z_n}(z-z_n)\begin{bmatrix} [\bbs M(z)]_{1,1} \\ [\bbs M(z)]_{2,1} \end{bmatrix} = \gamma_n\begin{bmatrix} [\bbs M(z_n)]_{1,2} \\ [\bbs M(z_n)]_{2,2} \end{bmatrix}.
\]
These relations now yield that
\[
\label{resM}
\sqrt{|\alpha_n|^2+|\beta_n|^2} =
\begin{cases}
|\gamma_n|\sqrt{S_2(z_n)}\leq |\gamma_n|, \medskip \\
\displaystyle \lim_{z\to z_n} |z-z_n|\sqrt{S_1(z)} \leq \lim_{z\to z_n} |z-z_n||B(z)|^{-1} = |B^\prime(z_n)|^{-1}.
 \end{cases}
\]
Recalling \eqref{reflectionless_sch}, we obtain the first bound of the proposition using the first estimate above, while the second bound follows from the second estimate above. Finally, the last bound is a consequence of Lemma~\ref{prop:1} and the second estimate above.
\end{proof}

\section{Results from Probability Theory}\label{app:probability_results}

In this Appendix we show in details some passages for the proof of Theorem~\ref{thm:local_CLT_intro}, and we report two results from Probability Theory that we used for the proof of Lemmas \ref{lem:CLT_approx} and \ref{lem:approx_envelope}. 

 \begin{lemma}
 \label{lem:scaling_max}
     Under Assumption \ref{hp:random_scattering_intro} and Definition \ref{def:subexp}, the following estimates hold:
     \begin{itemize}
         \item If $\alpha >0$ and $N>\exp\lp \frac{\nu^2}{2\alpha^2}\rp$, for all $s>0$ it holds that
    \[
         \mathbb{P}\lp \max_{i=1,\ldots,N} |\mu_i - \mu_{\mathcal D} | \geq 2\alpha \lp \ln(N) + s \rp \rp \leq 
             2e^{-s} \,.
     \]
        
     \item If $\alpha =0$, for all $s>0$ it holds that  
    \[
         \mathbb{P}\lp \max_{i=1,\ldots,N} |\mu_i - \mu_{\mathcal D} | \geq \nu \sqrt{2 \lp \ln(N) + s \rp} \rp \leq 2\ee^{-s} \,.
    \]
     \end{itemize}
 \end{lemma}

 \begin{proof}
 We prove only the first statement. The proof of the second one is analogous. Let $u= 2\alpha\lp\ln(N) + s\rp$. Then, from standard inequalities, 
\[
              \mathbb{P}\lp \max_{i=1,\ldots,N} |\mu_i - \mu_{\mathcal D}| \geq u \rp \leq \sum_{i=1}^N \mathbb{P}\lp |\mu_i - \mu_{\mathcal D}|  \geq u \rp \leq N \mathbb{P}\lp |\mu_1 - \mu_{\mathcal D} | \geq u \rp \,.
\]
Since $N>\exp\lp \frac{\nu^2}{2\alpha^2}\rp$,  we can apply \eqref{eq:subexp_prop} to conclude that
\[
    \mathbb{P}\lp \max_{i=1,\ldots,N} |\mu_i - \mu_{\mathcal D} |  \geq u \rp \leq 2 N \exp\lp -\frac{u}{2\alpha}\rp = 2 \ee^{-s}. \qedhere
\]   
 \end{proof}

From Lemma \ref{lem:scaling_max}
it follows that the set $\Omega$ defined in \eqref{eq:Omega} is exponentially small in probability as $N\to \infty$.

The next result is the so-called Nagaev--Guivarc’h method, which is a fundamental technique to prove probabilistic limit theorems for dynamical systems. We used the following theorem to prove Lemma \ref{lem:CLT_approx}, which is part of the proof of Theorem \ref{thm:local_CLT_intro}. 
\begin{theorem}[{\bf Nagaev--Guivarc’h method}, Theorem 4.2 in  \cite{Mazzuca2024}]\label{thm:Guido}
Let $X_1, X_2, \ldots$ be a sequence of real random variables and let $S_N := \sum_{j=1}^N X_j$. Assume that there exists $\delta>0$ and functions $\lambda(s,\xi) \in C^{1,0}([0,1) \times \R)$, $c_N(\xi) \in C^0(\R)$, and $h_n(\xi)$ continuous at zero, such that $\forall\ \xi \in [-\delta, \delta]$ and $\forall \ N \in \mathbb N$
\begin{equation}
\E \lb \ee^{- \ii \xi S_N }  \rb = c_N(\xi)\left(  \prod_{j=1}^N \lambda(\xi; \tfrac{j}{N})\right) \lp 1 + h_N(\xi)\rp \ .
\end{equation}
Moreover, assume that 
\begin{enumerate}
\item there exist functions $A, \sigma^2 :[0,1]\to \C$ such that 
\begin{equation}
\lambda(\xi;s) = \ee^{- \ii A(s) \xi  - \frac{\sigma^2(s)}{2}\xi^2 + o(\xi^2)} \ , \qquad \text{as } \xi \to 0\ ;
\end{equation}
\item $c_N(0)=1$ and  $\lim_{N\to \infty} c_N\left(\tfrac{\xi}{\sqrt{N}}\right) = \lim_{\xi\to \infty} c_N\left(\tfrac{\xi}{N}\right)=1$, $\forall \, \xi\in[-\delta,\delta]$;
\item  $h_N \to 0$ as $N\to \infty$, uniformly in $[-\delta,\delta]$ and $h_N(0)=0$.
\end{enumerate}
Then, $\int_0^1 A(s) \di s\in \R$ and $\int_0^1 \sigma^2(s) \di s \geq 0$ and 
\begin{equation}
\frac{S_N - N \int_0^1 A(s) \di s}{\sqrt{N}} \to \mathcal N \Big(0, \int_0^1\sigma^2(s) \di s \Big) \ , \qquad \text{as $N\to \infty$, in distribution.} 
\end{equation}
\end{theorem}

 Finally, we used a classical Probability Theory result, known as { Lyapounov's condition} \cite{Billingsley2012}, to prove Lemma \ref{lem:approx_envelope}, which is part of the proof of Theorem \ref{thm:envelope_intro}.
\begin{theorem}[\bf Lyapounov's condition]
\label{thm:lyapounov}
    Let $X_1,\ldots, X_N$ be independent random variables with means $\mu_1,\ldots,\mu_N$ and variances $\sigma^2_1,\ldots,\sigma^2_N$.
    Define $s_N := \sqrt{\sum_{k=1}^N\sigma^2_k}$ and assume that there exists a $\delta>0$ such that
    \begin{equation}
    \label{eq:lyapounouv_condition}
        \lim_{N\to\infty} \frac{1}{s_N^{2+\delta}} \sum_{k=1}^N \meanval{|X_k - \mu_k|^{2+\delta}}=0 \,.
    \end{equation}
    Then
\begin{equation}
    \frac{1}{s_N}\sum_{k=1}^N \lp X_k - \mu_k \rp \xrightarrow{d} \cN(0,1)\,  \quad \text{as }N\to\infty\,.
\end{equation}
\end{theorem}

\begin{proposition}[\bf Law of Large Numbers]
\label{rem:almost_sure}
   Under the same assumption of Theorem \ref{thm:lyapounov}, if $s_N$ is asymptotically equal to $c N^{\alpha}$ for some $c>0$ and  $0<\alpha < 1$, then it implies that 
    \begin{equation}
        \frac{1}{N}\sum_{k=1}^N X_k  - \frac{1}{N} \sum_{k=1}^N \mu_k \xrightarrow{N\to\infty} 0\,
    \end{equation}
   in probability. 
\end{proposition}

	\bibliographystyle{siam}
	\bibliography{mybib.bib}

\end{document}